\documentclass[11pt,onecolumn]{article}
\usepackage{times}
\usepackage{amsmath}
\usepackage{amssymb}
\usepackage{xspace}
\usepackage{theorem}
\usepackage{graphicx}
\usepackage{epsfig}
\usepackage{ifpdf}
\usepackage{url,hyperref}
\usepackage{latexsym}
\usepackage{euscript}
\usepackage{xspace}
\usepackage{color}
\usepackage{makeidx}
\usepackage{wrapfig,color}
\usepackage{stackrel}
\usepackage{amscd}
\usepackage[all]{xy}

\long\def\remove#1{}

\newtheorem{theorem}{Theorem}[section] % section

\newtheorem{claim}[theorem]{Claim}

\newtheorem{proposition}[theorem]{Proposition}
\newtheorem{definition}[theorem]{Definition}
\newenvironment{proof}{{\em Proof:}}{\hfill{\hfill\rule{2mm}{2mm}}}

\newcommand {\mm}[1] {\ifmmode{#1}\else{\mbox{\(#1\)}}\fi}

\newcommand{\eps}{{\varepsilon}}

\newcommand{\Z}			{\mathrm {\mathbb{Z}}}

\newcommand{\ann}	{{\sf a}}
\newcommand{\mysimeq}           {{\simeq}}

\newcommand{\rank}		{\mm {\rm rank}}
\newcommand{\Lk}			{\mathrm Lk\,}
\newcommand{\st}		{\mathrm {St\,}}

\newcommand{\G}		{\mathcal{G}}

\newcommand{\Rips}	{{\cal R}}

\newcommand{\Hom}	{{\rm Hom}}

\newcommand{\newfixpt}		{u}
\newcommand{\fv}		{{f_{\mathrm{V}}}} %{{f_{\mathbf{v}}}}

\newcommand{\id}		{{\mathrm id}}
\newcommand{\cst}[1]		{{\overline{\st \, #1}}}
\newcommand{\hh}		{{\hat{h}}}
\newcommand{\vmap}	{{\pi}}
\newcommand{\hv}		{{\hat{\vmap}}}
\newcommand{\hmap}	{{h}}

\begin{document}

\title{Computing Topological Persistence for Simplicial Maps}

\author{
Tamal K. Dey\thanks{
Department of Computer Science and Engineering,
The Ohio State University, Columbus, OH 43210, USA.
Email: {\tt tamaldey@cse.ohio-state.edu}}
\quad\quad
Fengtao Fan\thanks{
Department of Computer Science and Engineering,
The Ohio State University, Columbus, OH 43210, USA.
Email: {\tt fanf@cse.ohio-state.edu}}
\quad\quad Yusu Wang\thanks{
Department of Computer Science and Engineering,
The Ohio State University, Columbus, OH 43210, USA.
Email: {\tt yusu@cse.ohio-state.edu}}
}

\date{}
\maketitle

\begin{abstract}
Algorithms for persistent homology and zigzag persistent homology 
are well-studied for persistence modules where homomorphisms are induced 
by inclusion maps. 
In this paper, 
we propose a practical algorithm for computing persistence 
under $\mathbb{Z}_2$ coefficients 
for a sequence of general simplicial maps
and show how
these maps arise naturally in some applications of topological data analysis.

First, we observe that it is not hard to simulate 
simplicial maps by inclusion maps but not necessarily in a monotone direction. 
This, combined with the known algorithms for zigzag persistence, provides 
an algorithm for computing the persistence induced by simplicial maps. 

Our main result is that the above simple minded approach can be improved
for a sequence of simplicial maps given in a monotone direction.
A simplicial map can be decomposed into a set of elementary
inclusions and vertex collapses--two atomic operations
that can be supported efficiently with the notion of
simplex annotations for computing persistent homology.
A consistent annotation through these
atomic operations implies the maintenance of a
consistent cohomology basis, hence a homology basis by duality.
While the idea of maintaining a cohomology basis through
an inclusion is not new,
maintaining them through a vertex collapse is new, which constitutes an
important atomic operation for simulating simplicial maps.
Annotations support the vertex collapse in addition to the
usual inclusion quite naturally.

Finally, we exhibit an application
of this new tool in which we approximate the persistence 
diagram of a filtration of Rips complexes where 
vertex collapses are used to tame the blow-up in size. 

\end{abstract}

\section{Introduction}
\label{sec:intro}
Several applications in topological data analysis encounter the
following problem: when a simplicial
complex $K_1$ is modified to another complex $K_2$, how do the
topological features change. If the modification pertains only
to inclusions, that is, $K_1\subseteq K_2$, 
one can quantify the changes by the 
{\em persistent homology group}. 
This idea of topological persistence, originally introduced in~\cite{ELZ02},
has been explored extensively both algebraically and algorithmically in
the past decade, see e.g.~\cite{BD13,Carl09,CK11,CSM09,CEM06,EH09,NMS11,ZC05}. 
When the modification is more general than the inclusions, modeled
by considering the map
$K_1 \rightarrow K_2$ to be a simplicial map instead of
an inclusion map, the status is not the same. 
In this paper, we present an efficient algorithm for
computing topological persistence for simplicial maps and show its
application to a problem in topological data analysis.

%The idea of topological persistence along with its computation and application
%has been an active area of research for the past decade; see e.g~\cite{Carl09,EH09}.
%It was originally defined and explored 
%for nested sequence of simplicial
%complexes~\cite{ELZ02} which induced homomorphisms among their homology groups
%under inclusions. Later, further developments considered persistent
%homology and zigzag persistent homology for 
%general homomorphisms~\cite{CS10}.
%Although persistent homology and zigzag persistent homology under 
%inclusions have been investigated
%from algorithmic view point quite elaborately and 
%successfully~\cite{CK11,CSM09,CEM06,ELZ02,NMS11}, 
%the status is not the same 
%for general homomorphisms. 
Traditional persistent homology is defined for a \emph{monotone} sequence of homomorphisms, where all the maps $K_i \rightarrow K_{i+1}$ are along the same direction. 
In \cite{CS10}, Carlsson and de Silva introduced the \emph{zigzag persistence} defined for a zigzagging sequence of homomorphisms containing maps both of the form $K_i \rightarrow K_{i+1}$ and $K_i \leftarrow K_{i+1}$. 
They also presented a generic prototype algorithm for computing zigzag persistence induced by general homomorphisms. 
%In \cite{CS10}, Carlsson and de Silva presented a generic prototype algorithm for computing zigzag persistence induced by general homomorphisms.
It requires an explicit representation of the homomorphisms between the homology groups of two consecutive complexes in a sequence. 
In particular, if the input is given in terms
of maps between input spaces such as a continuous
map $f:K_i\rightarrow K_{i+1}$, a representation
of the induced homomorphism $f_*: H_*(K_i)\rightarrow H_*(K_{i+1})$ between
the homology groups needs to be computed.  
Often this step is costly and, in general, leads to $O(n^4)$ 
algorithm where each input complex has
$O(n)$ simplices. 
In contrast, when the map $f$ is an inclusion, the
persistence algorithm computes the persistent homology in
$O(n^3)$ time where $n$ is the total number of simplices inserted. 
%An $O(n^3)$ time algorithm is also developed for a zigzagging sequence of inclusions \cite{CSM09}. 

Using classical algebraic topological concepts such as mapping cylinders, 
it is not hard to simulate a simplicial map $f: K_i\rightarrow K_{i+1}$
by zigzag inclusions through
an intermediate complex $\hat{K}$ built from $K_i$. 
However, the complex $\hat{K}$, if constructed na\"{i}vely, 
may have a huge size. 
As detailed in section~\ref{sec:simulate},
one can improve upon this na\"{i}ve construction
which converts the input zigzag filtration connected by simplicial maps 
to another zigzag filtration connected only by inclusion maps.
One can then take advantage of the efficient algorithms to compute the 
persistence diagram for an inclusion-induced zigzag 
filtration \cite{CSM09,NMS11}.

Our main result detailed in Sections 
\ref{SEC:ANNOTATIONS} and \ref{SEC:ALGORITHM} 
is that when the input 
filtration is connected by a monotone (i.e, non-zigzag) 
sequence of simplicial maps, we can improve further upon the  
above construction by taking advantage of
annotations introduced recently in \cite{BCCDW12}.
%(also called homology signatures~\cite{EN11}) 
One of the main advantages of this approach is that it avoids the
detour through $\hat{K}$, and thus requires far fewer 
operations to move from $K_i$ to $K_{i+1}$; see Figure~\ref{fig:compare}.
Furthermore, the main auxiliary structure this new direct approach 
avails is  a set of
binary bits (elements of $\mathbb{Z}_2$)
attached to simplices which together can be 
viewed as a single binary matrix. 
This is in contrast to the simple-minded coning approach which uses the zigzag
persistence algorithm~\cite{CSM09} that requires
multiple such matrices.

One key aspect of our annotation based approach is that
it lets us simulate the simplicial maps by a sequence
of {\em inclusions} and {\em vertex collapses} in monotone direction 
{\em without} zigzag.
An annotation is linked with a cohomology basis which by duality
corresponds to a homology basis. Thus, annotations
over inclusions and vertex collapses allow us to maintain a consistent
homology basis indirectly under simplicial maps and 
infer the persistent homology. 
Our handling of inclusions can be seen as an
alternative formulation of the algorithm for computing persistent 
cohomology proposed in \cite{DMV11}. However, the handling of 
vertex collapses (which are neither inclusions nor deletions) 
in the context of persistence is new, and has not
been addressed previously.
%It has been observed that the algorithm proposed in \cite{DMV11} has 
%a very good practical performance in computing persistence (co)homology for non-zigzag inclusion-based filtrations \cite{DMV11b}.
%Our approach based on annotations extends the efficient maintenance 
%of a cohomology basis to the case when there are also collapse operations, 
%and hence we expect the resulting algorithm to be equally efficient in this more general setting. 

Finally, in Section~\ref{SEC:PDAPPROX}, we show 
an application where
the need for
computing persistence under simplicial maps arises naturally.
Our algorithm from Section \ref{SEC:ALGORITHM} can be used 
for this application directly. 
It is known that the persistence diagram~\cite{CEH07} of 
Vietoris-Rips (Rips in short)
filtrations provides avenues for topological analysis of 
%data~\cite{ALS11,CGOS09,CO08,DSW11}.
data~\cite{ALS11,DSW11,Ghrist}.
However, the inclusive nature of Rips complexes makes its size
too huge to be taken advantage of in practice. 
%To address this
%blow-up in size, Sheehy~\cite{Sheehy} suggested to sparsify the Rips
%complex by a sub-sampling. He used an elegant weighting
%scheme to replace the simplicial maps induced by vertex
%collapses with a sequence of inclusion maps, so that classical algorithm such as~\cite{ELZ02}
%can be used. 
One can consider sparsified versions of
Rips complexes~\cite{Sheehy} or 
graph induced complexes~\cite{DFW13} by subsampling input points
which can be achieved by vertex collapses.
Our algorithm supports vertex collapses and thus
naturally yields to maps arising out of such subsampling.
%obviates the need for the weighting scheme.
%Recently, Sheehy~\cite{Sheehy} showed a novel method to approximate the persistence diagram of a Rips filtration
%from another filtration where vertex deletons compensate for the blow-up 
%in the number of simplices. The main observation is that a weighting scheme along with a sophisticated data
%structure called net-tree~\cite{} can replace the simplicial maps induced by vertex collapses
%with inclusions at the homology level. Our algorithm for computing persistences of simplicial maps 
%obviates the necessity of the weighting scheme and the net-tree data structure.

Throughout the paper, simplicial homology and
cohomology groups are defined with 
coefficients in $\mathbb{Z}_2$. 

\section{Preliminaries and simplicial maps}
\label{SEC:SIMPLICIALMAPS}
%\vspace*{0.1in}
\begin{definition}
Given a finite set $V$, a simplicial complex $K=K(V)$ is defined as a collection
of subsets $\{\sigma \subseteq V\}$ so that $\sigma\in K$ implies that
any subset $\sigma'\subseteq \sigma$ is in $K$. The vertex
set $V(K)$ of $K$ is $V$. The elements of $K$ 
are called its simplices. An element $\sigma\in K$
is a $p$-simplex if its cardinality is $p+1$. A simplex 
$\sigma'$ is a face of $\sigma$ and $\sigma$ is a coface of
$\sigma'$ if $\sigma'\subseteq \sigma$.
\end{definition}
 
\begin{definition}
Let $X$ be a subset of a simplicial complex $K$. The set
$\st X:=\{\sigma'\in K \,|\, \mbox{ $\exists \sigma\in X$ and $\sigma\subseteq \sigma'$}\}$
is called the star of $X$. The closure of $X$, denoted $\overline{X}$,
is the simplicial complex formed by simplices in $X$ and all of their
faces. The link of $X$ is
$ \Lk X := \cst X \setminus \st \overline{X} .$
\end{definition}

The star of $X$ consists of the simplices in $K$ 
that are cofaces of simplices in $X$. 
The link of $X$ consists of the faces of the simplices in its star which 
contain no vertex of $X$.
\subsection{(Co)homology groups}
We briefly introduce the notion of homology and cohomology groups here
which we use extensively; see e.g. Hatcher~\cite{Hatcher} for details.
Both groups are defined under $\Z_2$ coefficients.
A $p$-{\em chain} $c_p$ in a simplicial complex $K$ is a
formal sum of $p$-simplices, 
that is, $c_p=\Sigma\alpha_i\sigma_i$, $\alpha_i\in\{0,1\}$ 
and $\sigma_i\in K$. 
The chains under $\Z_2$-additions form an abelian group called
the $p$-chain group of $K$ and is denoted $C_p(K)$.
The boundary of a $p$-simplex $\sigma$, denoted $\partial_p \sigma$,
is defined to be the formal sum of its boundary $(p-1)$-simplices. We obtain 
the boundary homomorphism 
$\partial_p:C_p\rightarrow C_{p-1}$ given by
$\partial_p(\Sigma\alpha_i\sigma_i)=\Sigma \alpha_i(\partial_p\sigma_i)$.
The kernel of $\partial_p$ is 
the {\em cycle} group $Z_p\subseteq C_p$. The image of $\partial_p$ is
the boundary group $B_{p-1}\subseteq C_{p-1}$. It can be easily
verified that $\partial_p\circ \partial_{p-1}=0$ which makes 
the quotient group $H_p(K)=Z_p(K)/B_p(K)$, known as the $p$th homology
group, well defined.

Cohomology groups are defined by cochains, cocycles, and coboundaries
that are, in a sense, functional duals to the chains, cycles, 
and boundaries respectively.
A $p$-cochain is a homomorphism $c^p: C_p(K)\rightarrow \Z_2$ and thus
can be completely specified by its value on each $p$-simplex.
The $p$-cochain group $C^p(K)$ is the group of all cochains under
$\Z_2$-additions. The coboundary operator $\delta_p: C^{p}\rightarrow C^{p+1}$
sends $p$-cochains to $(p+1)$-cochains by evaluating $\delta_pc^p$ on
each chain $d_{p+1}\in C_{p+1}$ as $c^p(\partial_{p+1} d_{p+1})$. The kernel
of $\delta_p$ is the cocycle group $Z^p(K)$ and its image is the
coboundary group $B^{p+1}(K)$. Since $\delta_p\circ \delta_{p+1}=0$, we
have the quotient group $Z^p(K)/B^p(K)$ well defined which is called the
cohomology group $H^p(K)$.

\subsection{Simplicial maps}
\label{sec:simp-map}
\begin{definition}
A map $f: K \rightarrow K'$ is \emph{simplicial} if for every
simplex $\sigma=\{v_0,v_1,\ldots,v_k\}$ in $K$, 
$f(\sigma)=\{f(v_0),f(v_1),$ $\ldots,f(v_k)\}$ is a simplex in $K'$. 
The restriction $f_V$ of $f$ to $V(K)$ is a vertex map.
\end{definition}

A simplicial map $f: K_1\rightarrow K_2$
induces a homomorphism 
$H_p(K_1)\stackrel{f_*}{\rightarrow} H_p(K_2)$ for the 
homology groups in the forward
direction while a homomorphism $H^p(K_1)\stackrel{f^*}{\leftarrow} H^p(K_2)$ 
for the cohomology groups in the backward direction. The latter sends a 
cohomology class $[c]$ in $H^p(K_2)$ to the cohomology class
$[c']$ in $H^p(K_1)$ where $c'(c_p)= c(f(c_p))$ for each $c_p\in C_p(K_1)$.

\begin{definition}
A simplicial map $f: K\rightarrow K'$ is called {\em elementary} if
it is of one of the following two types: 
%the induced vertex map $\fv$ is injective everywhere except possibly
%on a set $X\subseteq V(K)$, $|X|>1$, for which $\fv(X)$ is a single vertex  in $K'$. 
%is a single vertex in $K'$. 
%$\fv(\fixpt)$ in $K'$ with $\fixpt \in X$; $\fixpt$, if exists, is called a \emph{fixed point} of $f$.  
\begin{itemize}
\item $f$ is injective, and 
$K'$ has at most one more simplex than $K$. In this case,
$f$ is called an {\em elementary inclusion}.
\item $f$ is not injective but is surjective, and the vertex
map $f_V$ is injective everywhere
except on a pair $\{u,v\}\subseteq V(K)$. In this case, $f$ is called an 
{\em elementary collapse}. An elementary collapse maps 
a pair of vertices into a single
vertex, and is injective on every other vertex.
\end{itemize}
\label{def-esimp}
\end{definition}
We observe that any simplicial map is a composition
of elementary simplicial maps (see Appendix~\ref{appendix:A}). 
\begin{proposition}
If $f: K\rightarrow K'$ is a simplicial map, then there are
elementary simplicial maps $f_i$ 
$$
K\stackrel{f_1}{\rightarrow} K_1\stackrel{f_2}{\rightarrow} K_2\cdots
\stackrel{f_n}{\rightarrow} K_n=K'
\mbox{ so that } f=f_n\circ f_{n-1}\circ\cdots \circ f_1. 
$$ 
\label{elementary}
\end{proposition}
In view of Proposition~\ref{elementary}, it is sufficient to show how one
can design the persistence algorithm for an elementary simplicial
map. At this point, we make a change in the
definition~\ref{def-esimp} of elementary simplicial maps that eases further discussions. We let $\fv$ to be identity
(which is an injective map) everywhere except possibly on a 
pair of vertices $\{u,v\} \subseteq V(K)$ for
which $\fv$ maps to a single vertex, say $u$ in $K'$. 
This change can be implemented 
by renaming the vertices in $K'$ that are
mapped onto injectively. 
Since the standard persistence algorithm handles inclusions,
we focus mainly on handling the elementary collapses.

\subsection{Simulation with coning}
\label{sec:simulate}
First, we propose a simulation of simplicial maps with
a coning strategy that only requires additions of simplices.
\begin{figure}[h!]
\begin{center}
\includegraphics[width=0.45\textwidth]{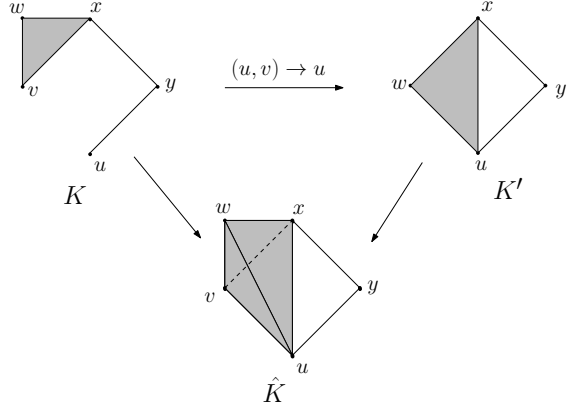}
\end{center}
\caption{Elementary collapse $(u,v)\rightarrow u$: the cone
$u*\cst v$ adds edges $\{u,w\}, \{u,v\}, \{u,x\}$, triangles
$\{u,w,x\}, \{u,v,x\},\{u,v,w\}$, and the tetrahedron
$\{u,v,w,x\}$.}
\label{star-fig}
\end{figure}
%\end{wrapfigure}
We focus on elementary collapses.
Let $f: K\rightarrow K'$ be an elementary collapse.
Assume that the induced vertex map collapses vertices $u,v \in K$ to $u \in K'$, and is identity on other vertices.
For a subcomplex $X\subseteq K$, define the cone $u*X$ to be the
complex $\{\overline{\sigma\cup \{u\}}\,| \, \sigma \in X\}$.
Consider the augmented complex
$$\hat{K} := K \cup \left(u * \cst v\right).$$
In other words, for every simplex $\{u_0, \ldots, u_d\} \in \cst v$ of $K$, we add the simplex  $\{u_0, \ldots, u_d \} \cup \{ u \}$ to $\hat{K}$ if it is not already in.
See Figure~\ref{star-fig}.
Notice that $K'$ is a subcomplex of
$\hat{K}$ in this example which we observe is true in general.

\begin{claim}
$K' \subseteq \hat{K}$.
\label{easy-claim}
\end{claim}
\begin{proof}
For a simplex $\sigma \in K'$ that does not contain $\newfixpt$, $f$ is identity on its unique pre-image; that is, $\sigma \in K \subseteq \hat{K}$.
Now consider a $d$-simplex $\sigma=\{\newfixpt, u_1, \ldots, u_d \}\in K'$.
Since $f$ is surjective, there exists at least one pre-image of $\sigma$ in $K$ of the form $\sigma'=\{u_0, u_1, \ldots, u_d\}$, where  $u_0$ is either $u$ or $v$. If it is $u_0=u$, we have $f(\sigma')=\sigma'=\sigma$ and thus $\sigma\in K\subseteq \hat{K}$. So, assume that
$u_0=v$.
This means that the simplex $\{u_1, \ldots, u_d\}$ is in $\Lk v$ (and thus in $\cst v$), implying that $\sigma=\{\newfixpt, u_1, \ldots, u_d\} \in \hat{K}$.
\end{proof}

\begin{wrapfigure}{r}{1in}
\input{commute.pstex_t}
\end{wrapfigure}
Now consider the canonical inclusions $i:K \hookrightarrow \hat{K}$
and $i': K'\hookrightarrow \hat{K}$. These inclusions constitute the
diagram on the righthand side which does not necessarily commute.
Nevertheless, it commutes at the homology level which
is precisely stated below.\\

\begin{table*}
\begin{center}
\begin{tabular}{c}
$$
\xymatrix
{
H_*(K_1) \ar@{->}[r]^{f_{1_*}} \ar@{->}[d]^{=}
& \  H_*(K_2) \  \ar@{->}[r]^{=}\ar@{->}[d]^{\mysimeq}
& \ H_*(K_2) \  \ar@{->}[r]^{=}\ar@{->}[d]^{=}
& \ H_*(K_2) \   \ar@{<-}[r]^{f_{2_*}}\ar@{->}[d]^{\mysimeq}
& \ H_*(K_3) \  \ar@{->}[r]^{f_{3_*}}\ar@{->}[d]^{=}
& \ \ldots \ldots \ar@{->}[r] \
& \  H_*(K_m) \ar@{->}[d]^{=}
\\
H_*(K_1) \ar[r]^{i_{1_*}}
& \ H_*(\hat{K}_1) \  \ar[r]^{\mysimeq}
& \ H_*(K_2)\  \ar[r]^{\mysimeq}
& \ H_*(\hat{K}_3)\   \ar@{<-}[r]^{i_{2_*}}
& \ H_*(K_3)\ \ar@{->}[r]^{i_{3_*}}
& \ \ldots \ldots\  \ar@{->}[r]
& \ H_*(K_m)
}
$$
\end{tabular}
\label{comm:diag}
\end{center}
\end{table*}
\begin{proposition}
$f_*: H_*(K)\rightarrow H_*(K')$ is equal to
$(i'_*)^{-1}\circ i_*$
where $i'_*$ is an
isomorphism and $H_*(K)\stackrel{i_*}{\rightarrow} H_*(\hat{K})
\stackrel{i'_*}{\leftarrow}H_*(K')$.
\label{zigzag}
\end{proposition}
\begin{proof}
We use the
notion of contiguous maps which induces equal maps
at the homology level. Two maps $f_1:K_1\rightarrow K_2$,
$f_2: K_1\rightarrow K_2$ are contiguous if for every simplex
$\sigma\in K_1$, $f_1(\sigma)\cup f_2(\sigma)$ is a
simplex in $K_2$. We observe that the simplicial maps
$i'\circ f$ and $i$ are contiguous and
$i'$ induces an isomorphism at the homology level,
that is, $i'_*: H_*(K)\rightarrow H_*(\hat{K})$
is an isomorphism.

Since $i$ is contiguous to $i'\circ f$ (Proposition A.1 in appendix),
we have $i_*=(i'\circ f)_*= i'_* \circ f_*$. Since $i'_*$ is
an isomorphism (Proposition A.2 in appendix),
$(i'_*)^{-1}$ exists and is an isomorphism. It then follows that
$f_*=(i'_*)^{-1} \circ i_*$.
\end{proof}
~\\

Proposition~\ref{zigzag} allows us to simulate the persistence of
a sequence of simplicial maps with only inclusion-induced homomorphisms.
%\subsection{Zigzag persistence of simplicial maps}
%\label{sec:zigzag}
%\vspace*{0.1in}
Consider the following sequence of simplicial complexes connected with
a zigzag sequence of simplicial maps (the arrows could be
oriented differently in general):
%$$
%K_1 \stackrel{f_1}{\leftrightarrow} K_2
%\stackrel{f_2}{\leftrightarrow}\cdots \stackrel{f_{n-1}}{\leftrightarrow}K_n
%$$
%where the arrows $\leftrightarrow$ can either be a forward arrow $\rightarrow$
%or a backward arrow $\leftarrow$ resulting into a {\em zigzag sequence}.
%For specificity, assume that we have a zigzag sequence:
$$
K_1 \stackrel{f_1}{\rightarrow} K_2
\stackrel{f_2}{\leftarrow} K_3 \stackrel{f_3}{\rightarrow}\cdots
\rightarrow K_n
$$
which generates the module at the homology level by induced
homomorphisms ${f_i}_*$
$$
{\cal F}:
H_*(K_1) \stackrel{{f_1}_*}{\rightarrow} H_*(K_2)
\stackrel{{f_2}_*}{\leftarrow} H_*(K_3) \stackrel{{f_3}_*}{\rightarrow}\cdots
\rightarrow H_*(K_n)
$$
When the maps $f_i$ are all inclusions,
it is known that the zigzag persistence induced by them can be computed
in matrix multiplication time by a recent algorithm of~\cite{NMS11}. 
This algorithm does not extend to simplicial maps as per se though
we know that a persistence module induced by simplicial maps
admits a decomposition~\cite{ZC05} and hence a persistence diagram
~\cite{CEH07}. With our
observation that every simplicial map can be simulated with inclusion
maps, we can take advantage of the algorithm of~\cite{NMS11} for
computing zigzag persistence for simplicial maps.
In view of Proposition~\ref{zigzag}, consider the following sequence
connected only with inclusions:
$$
K_1 \hookrightarrow \hat{K_1} \hookleftarrow K_2
\hookrightarrow \hat{K_3} \hookleftarrow K_3\hookrightarrow\cdots
\hookleftarrow K_n
$$

At the homology level we have $H_*(\hat{K}_i)\mysimeq H_*(K_{i+1})$
induced by the inclusion $\hat{K}_i\hookleftarrow K_{i+1}$ and also
$H_*(K_i)\mysimeq H_*(\hat{K}_{i+1})$ induced by the inclusion
$K_i \hookrightarrow \hat{K}_{i+1}$.
Thus, we have the following persistence module:
%Therefore, the zigzag persistence of ${\cal F}$ can be read from the persistence of the module
\begin{eqnarray*}
{\cal M}: H_*(K_1)& \stackrel{i_*}{\rightarrow}& H_*(\hat{K_1})\\
& \mysimeq& H_*(K_2)\\
& \mysimeq& H_*(\hat{K_3})\stackrel{i_*}{\leftarrow} H_*(K_3)\stackrel{i_*}{\rightarrow}
\cdots
\stackrel{i_*}{\leftarrow} H_*(K_n)
\end{eqnarray*}

\begin{theorem}
The persistence diagram of ${\cal F}$ can be derived from the
that of the module ${\cal M}$.
\label{thm:FM}
\end{theorem}
\begin{proof}
Consider the diagram between vector spaces as shown above.
All isomorphisms are induced by inclusions, hence every square
being supported only by isomorphisms commutes.
The other squares supported by $f_{i_*}$ also
commute because of Proposition~\ref{zigzag}.
Hence every square in this diagram commutes, and
the claim follows~\cite{CS10,EH09}.
\end{proof}

%\section{Simplicial maps and annotations}
\section{Annotations}
\label{SEC:ANNOTATIONS}
When we are given a \emph{non-zigzag} sequence of 
simplicial maps
$
K_1\stackrel{f_1}{\rightarrow} K_2\stackrel{f_2}{\rightarrow} K_3\cdots
\stackrel{f_n}{\rightarrow} K_n
$
we can improve upon the coning approach by reducing simplex
insertions as illustrated in Figure~\ref{fig:compare}. 
Consider the map $f_{ij}: K_i \rightarrow K_j$ where
$f_{ij}= f_{j-1}\circ \cdots \circ f_{i+1}\circ f_i$. 
%The persistent homology
%is given by the image of the induced homomorphism
%$f_{ij_*}: H_*(K_i)\rightarrow H_*(K_j)$. 
To compute the
persistent homology, the persistence algorithm essentially
maintains a consistent basis by computing the
image $f_{ij_*}(B_i)$ of a basis
$B_i$ of $H_*(K_i)$. 
As one moves through a map in the 
filtration, the homology basis elements get
created (birth) or can be interpreted to be destroyed (death).
The notion of this birth and death 
of the homology basis elements can be formulated precisely
with algebra~\cite{ZC05} and can be summarized with 
persistence diagrams~\cite{CEH07}.
Here, instead of a consistent homology basis, we maintain a consistent
cohomology basis, that is, if $B^i$ is a cohomology basis
of $H^*(K_i)$ maintained by the algorithm, we compute
the preimage $f_{ij}^{*-1}(B^i)$ where $H^*(K_i)
\stackrel{f_{ij}^*}{\leftarrow} H^*(K_j)$
is the homomorphism induced in the cohomology groups by $f_{ij}$.
By duality, this implicitly maintains a consistent homology basis
and thus captures all information about persistent homology 
as well~\cite{DMV11}.

Our main tool to maintain a consistent cohomology basis is the
notion of annotation~\cite{BCCDW12} which 
are binary vectors assigned to simplices.
We maintain the annotations as we go
forward through the given sequence, and thus maintain a cohomology
basis in the reverse direction whose
birth and death coincide with the death and birth respectively of a 
consistent homology basis. 

\begin{definition}
Given a simplicial complex $K$, Let $K(p)$ denote the set of
$p$-simplices in $K$. An annotation for $K(p)$ is an
assignment $\ann: K(p)\rightarrow \mathbb{Z}_2^g$
of a binary vector $\ann_{\sigma}=\ann(\sigma)$ of same length $g$
for each $p$-simplex $\sigma \in K$. Entries
of $\ann_\sigma$ are called its elements. We also have an induced
annotation for any $p$-chain $c_p$ given by
$\ann_{c_p}= \Sigma_{\sigma\in c_p} \ann_{\sigma}$.
\end{definition}

\begin{definition}
An annotation $\ann:K(p)\rightarrow \mathbb{Z}_2^g$ is {\em valid} if 
conditions 1 and 2 are satisfied:
\begin{enumerate}
\item $g=\rank\, H_p(K)$, and
\item two $p$-cycles $z_1$ and $z_2$  have 
$\ann_{z_1}=\ann_{z_2}$ iff their homology
classes are identical, i.e. $[z_1]=[z_2]$.
\end{enumerate}
\end{definition}

%In~\cite{BCCDW12}, it was shown that valid annotations are 
%computable efficiently.

\begin{proposition}
Statements 1 and 2 are equivalent: 
\begin{enumerate}
\item An annotation $\ann: K(p)\rightarrow \mathbb{Z}_2^g$ 
is valid
\item The cochains $\{\phi_i\}_{i=1,\cdots,g}$
given by $\phi_i(\sigma) = \ann_{\sigma}[i]$ for all $\sigma\in K(p)$
are cocycles whose cohomology classes $\{[\phi_i]\},
i=1,\ldots,g$ constitute a basis of $H^p(K)$.
\end{enumerate}
\label{annot-cohom}
\end{proposition}
\begin{proof}
$1\rightarrow 2$: The cochains $\phi_i$ are cocycles since for any
$(p+1)$-simplex $\tau\in K(p+1)$ one has 
$[\partial \tau]=[0]$ and hence $\delta_p\phi_i(\tau) = \phi_i(\partial\tau)=\phi_i(0)=0$, where $\delta_p$ is the co-boundary operator for $p$-dimensional co-chains. 
Let $[z_1],[z_2],\cdots,[z_g]$ be a basis of $H_p(K)$. Let $V$ be the
vector space generated by $[\phi_i]$, $i=1,\cdots,g$. Define a 
bilinear form $\alpha: V\times H_p(K)\rightarrow \mathbb{Z}_2$ by
$\alpha([\phi_i],[z_j])=\phi_i(z_j)$. The matrix
$[\phi_i(z_j)]_{ij}$ has full rank due to the condition 2 in the definition
of annotation. This means the vector spaces $V$ and $H_p(K)$ have
the same rank and hence are isomorphic. It follows that $V\mysimeq H^p(K)$.

$2\rightarrow 1$: For this direction, consider a basis $[z_1],[z_2],\cdots,[z_g]$ of
$H_p(K)$. By universal coefficient theorem we have an isomorphism
$H^p(K)\mysimeq \Hom (H_p(K),\Z_2)$ which sends a cocycle
class 
$[\phi_i]$ to the homomorphism
$[z_j]\mapsto \phi_i(z_j)$.
%and a bilinear form $\alpha: H^p(K)\times H_p(K)$ as above.
This means that the matrix $[\phi_i(z_j)]_{ij}$
has full rank and hence the vectors $[\phi_1(z_j),\ldots,\phi_g(z_j)]$
and $[\phi_1(z_k),\ldots,\phi_g(z_k)]$ are identical if and only if 
$[z_i]=[z_k]$. 
The claim can be extended to any homology
class since it can be expressed as a linear
combination of the basis elements. 
\end{proof}
~\\
 
In light of the above result, an annotation is simply one way to represent 
a cohomology basis. 
However, by representing the corresponding basis as an explicit 
vector associated with each simplex, it localizes the basis to each simplex. 
As a result, we can update the cohomology basis locally by changing the annotations locally (see Proposition \ref{annot-push}). 
This point of view also helps to reveal how we can process elementary collapses, which are neither inclusions nor deletions, by transferring 
annotations (see Proposition \ref{claim:transferresult}). 

%In our algorithm, elementary collapses need to modify annotations 
%in the current simplicial complex. The result below helps us to 
%transfer annotations among simplices without affecting cohomology bases. 

\begin{figure*}[ht!]
\centering
\input{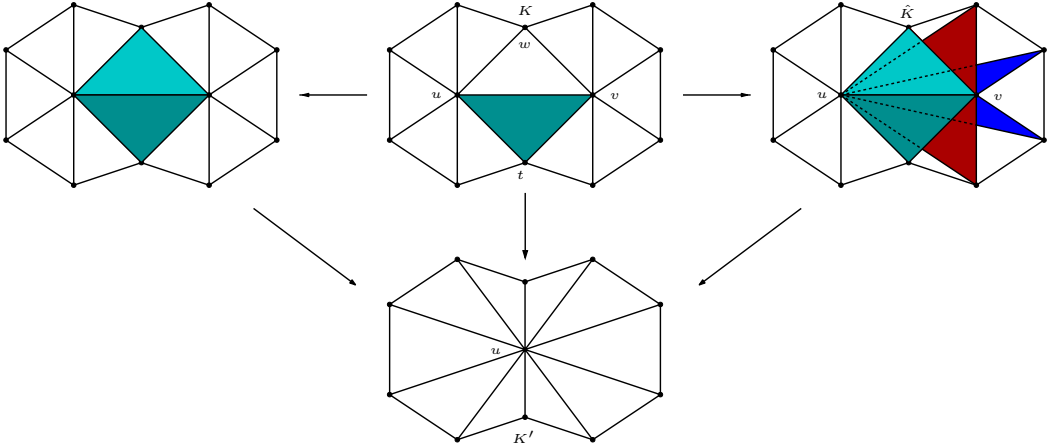}
\caption{Annotation vs. coning: The pair $(u,v)$ is collapsed
to $u$ in $K$ to produce $K'$ (middle column). The $2$-simplices are the shaded triangles
alone. Annotation requires inserting (implicitly)
the single triangle as shown on the left
whereas coning requires inserting many more simplices as shown by shaded triangles on the right.
Specifically, the coning approach requires inserting all simplices in the cone $u * \cst{v}$ formed by $u$ and all simplices in the closure of the star of $v$.}
\label{fig:compare}
\end{figure*}

\section{Algorithm}
\label{SEC:ALGORITHM}
Consider the persistence module
${\cal M}$ induced by elementary
simplicial maps $f_i: K_i\rightarrow K_{i+1}$.
$$
{\cal M}: \,\, H_*(K_1)\stackrel{f_{1_*}}{\rightarrow} H_*(K_2)\stackrel
{f_{2_*}}{\rightarrow} H_*(K_3)\cdots
\stackrel{f_{n_*}}{\rightarrow} H_*(K_n)
$$
Instead of tracking a consistent homology basis for the module ${\cal M}$,
we track a cohomology basis in the dual module ${\cal M}^*$
where the homomorphisms are in reverse direction:
$$
{\cal M}^*: \,\, H^*(K_1)\stackrel{f^*_1}{\leftarrow} H^*(K_2)\stackrel
{f^*_{2}}{\leftarrow} H^*(K_3)\cdots
\stackrel{f^*_{n}}{\leftarrow} H^*(K_n)
$$
%It is known that the persistence diagram of ${\cal M}$ and ${\cal M}^*$
%are identical~\cite{DMV11}. Therefore, instead of tracking a homology basis,
%we can track a cohomology basis for ${\cal M}^*$ which we precisely
%achieve by annotations.
As we move from left to right in the above sequence, 
the annotations implicitly maintain a cohomology basis
whose elements are also {\em time stamped} to signify when a basis 
element is born or dies.
We should keep in mind that the \emph{birth} and \emph{death} of a cohomology basis element
coincides with the \emph{death} and \emph{birth} of a homology basis element because the
two modules run in opposite directions.

\subsection{Elementary inclusion} 
\begin{figure*}[ht!]
\begin{center}
\begin{tabular}{c||c}
\includegraphics[width=0.36\textwidth]{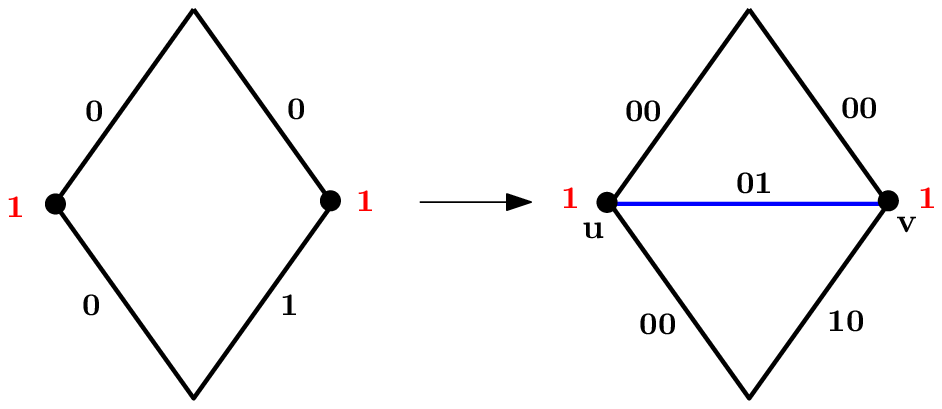}&
\includegraphics[width=0.54\textwidth]{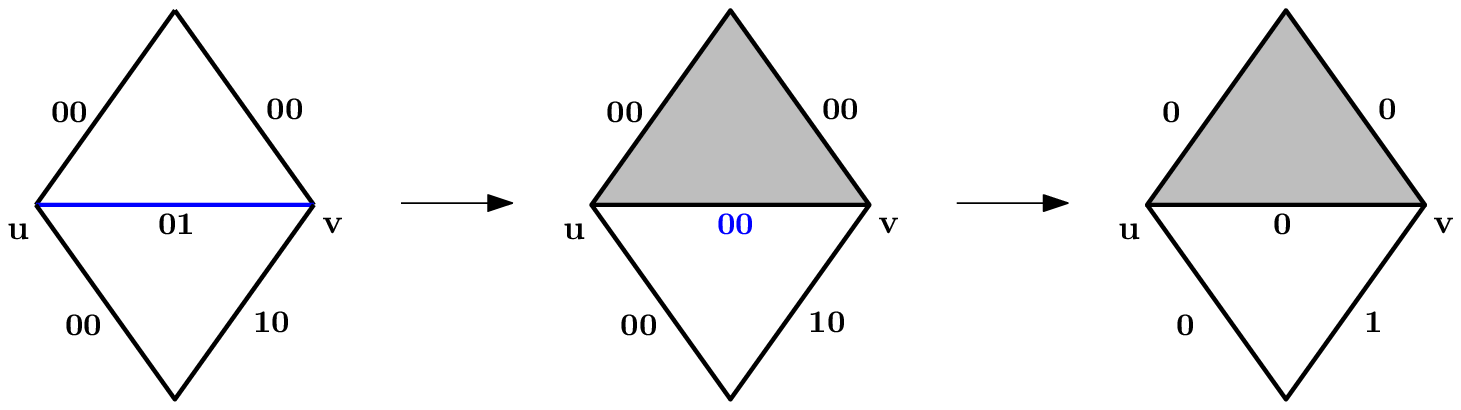}\\
(a) Case(i) & (b) Case(ii)
\end{tabular}
\end{center}
\caption{Case(i) of inclusion: the boundary $\partial uv=u+v$ of the edge
$uv$ has annotation $1+1=0$. After its addition,
every edge gains an element in its annotation which is $0$ for all 
except the edge $uv$.
Case (ii) of inclusion: the boundary of the top triangle has annotation
$01$. It is added to the annotation of $uv$ which is the only edge having
the second element $1$. Consequently the second 
element is zeroed out for every edge, and
is deleted.}
\label{inclusion-annot}
\end{figure*}
The handling of elementary inclusions using annotations can be 
viewed as an alternative formulation of the algorithm 
proposed in \cite{DMV11}; 
%with the annotations corresponding to the optimization used in that algorithm in maintaining basis; 
see also \cite{DMV11b}. 
We describe it in terms of the annotation here 
because it is also used in an elementary collapse, a new atomic
operation that we need to address. %of presentation. 
Consider an elementary inclusion $K_i \hookrightarrow K_{i+1}$.
Assume that $K_i$ has a valid annotation. We describe how we obtain
a valid annotation for $K_{i+1}$ from that of $K_i$ after
inserting the $p$-simplex 
$\sigma= K_{i+1}\setminus K_i$. We compute the
annotation $\ann_{\partial \sigma}$ for the boundary $\partial\sigma$
in $K_i$ and take actions as follows. A formal justification
is provided in Section~\ref{SEC:JUSTIFICATION}.\\

\noindent
Case (i): If $\ann_{\partial\sigma}$ is a zero vector, the class
$[\partial\sigma]$
is trivial in $H_{p-1}(K_i)$. This means 
$\sigma$ creates a $p$-cycle in $K_{i+1}$ and by duality a $p$-cocycle
is killed while going left from $K_{i+1}$ to $K_i$. 
In this case we augment the annotations for all $p$-simplices by one
element with a time stamp $i+1$, that is, an
annotation $[b_1,b_2,\cdots,b_g]$ for a $p$-simplex $\tau$ is updated to
$[b_1,b_2,\cdots,b_g,b_{g+1}]$ with the last element time stamped $i+1$ where
$b_{g+1}=0$ for $\tau\not= \sigma$ and
$b_{g+1}=1$ for $\tau=\sigma$. 
The element $b_i$ of $\ann_\sigma$ is set to zero for
$1\leq i \leq g$.
 Other annotations for other simplices remain unchanged. See Figure~\ref{inclusion-annot}(a).\\

\noindent
Case (ii): If $\ann_{\partial\sigma}$ is not a zero vector,
the class of the $(p-1)$-cycle $\partial\sigma$ is nontrivial in 
$H_{p-1}(K_i)$. Therefore, $\sigma$ kills the class of this cycle and
a corresponding dual class of cocycles is born in the
reverse direction. We simulate
it by forcing $\ann_{\partial\sigma}$ to be zero which affects
other annotations as well. Let $i_1, i_2,\cdots,i_k=u$ be the set of
indices in non-decreasing order so that $b_{i_1},b_{i_2},\cdots,b_{i_k}=b_u$
are all of the nonzero elements 
in $\ann_{\partial\sigma}=[b_1,b_2,\cdots, b_u,\cdots, b_g]$.
The cocycle $\phi=\phi_{i_1}+\phi_{i_2}+\cdots+(\phi_{i_k}=\phi_u)$ 
should become a coboundary 
after the addition of $\sigma$, which renders 
$$\phi_u=\phi_{i_1}+\phi_{i_2}+\cdots+\phi_{i_{k-1}}.$$ 
We make the latest cocycle
$\phi_u$ to be dependent on others. In other words, the cocycle class $[\phi]$
which is born at the time $i+1$ is chosen to be killed at time when $b_u$ was
introduced.
%%%%%%added by fengtao%%%%%%
This pairing matches that of the standard persistence algorithm 
where the youngest basis element is always paired first.
%%%%%%%%%%
%To simulate the purport of the above equation, we replace the
%class $[\phi_u]$ with the class 
%$[\phi_{i_1}+\phi_{i_2}+\cdots+\phi_{i_{k-1}}]$ and
%delete the class $[\phi_u]$ from the cocycle basis.
We add the vector $\ann_{\partial\sigma}$ 
to all annotations of $(p-1)$-simplices whose $u$th element is nonzero.
This zeroes out the $u$th element of all annotations of $(p-1)$-simplices.
We simply delete this element from all such annotations. 
See Figure~\ref{inclusion-annot}(b).

Notice that determining if we have case (i) or (ii) can be done 
easily in $O(pg)$ time by checking the annotation of $\partial \sigma$. 
Indeed, this is achieved because the annotation already localizes 
the co-homology basis to each individual simplex. 

\subsection{Elementary collapse}
The case for handling collapse is more interesting. 
It has three distinct steps, (i) 
elementary inclusions to satisfy the so called link
condition, (ii) local annotation transfer to prepare
for the collapse, and (iii) collapse of the simplices with
updated annotations. We explain each of these steps now.

The elementary inclusions that may precede the final collapse
are motivated by a result that connects collapses with the change
in (co)homology.
Consider an elementary collapse  $K_i\stackrel{f_i}{\rightarrow} K_{i+1}$
where the vertex pair $(u,v)$ collapses to $u$. 
The following link condition, introduced in~\cite{DEGN99} and later
used to preserve homotopy~\cite{ALS11}, becomes relevant.
\begin{definition}
A vertex pair $(u,v)$ in a simplicial
complex $K_i$  satisfies the {\em link condition} if the edge
$\{u,v\}\in K_i$, and $\Lk u \cap \Lk v= \Lk \{u,v\}$.
An elementary collapse $f_i: K_i\rightarrow K_{i+1}$
satisfies the link condition if the vertex pair on which
$f_i$ is not injective satisfies the link condition.
\end{definition}

\begin{figure*}[ht!]
\begin{center}
\begin{tabular}{c}
\includegraphics[width=0.59\textwidth]{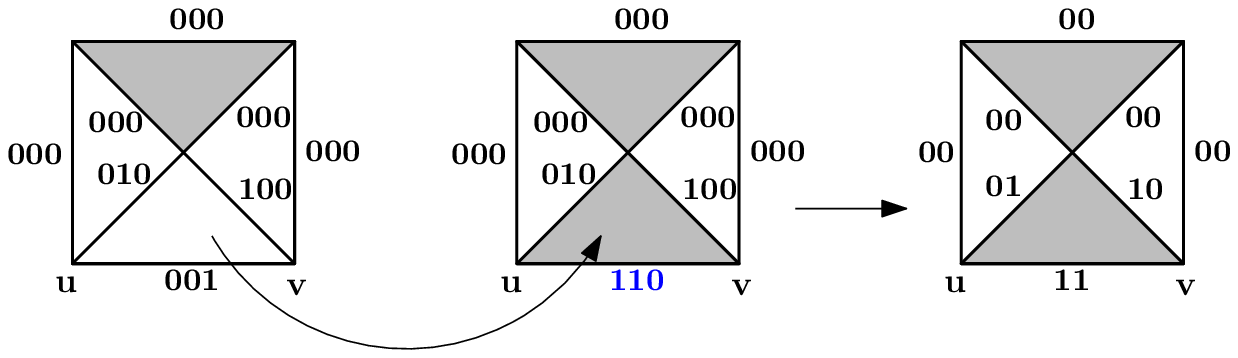}\\
\includegraphics[width=0.58\textwidth]{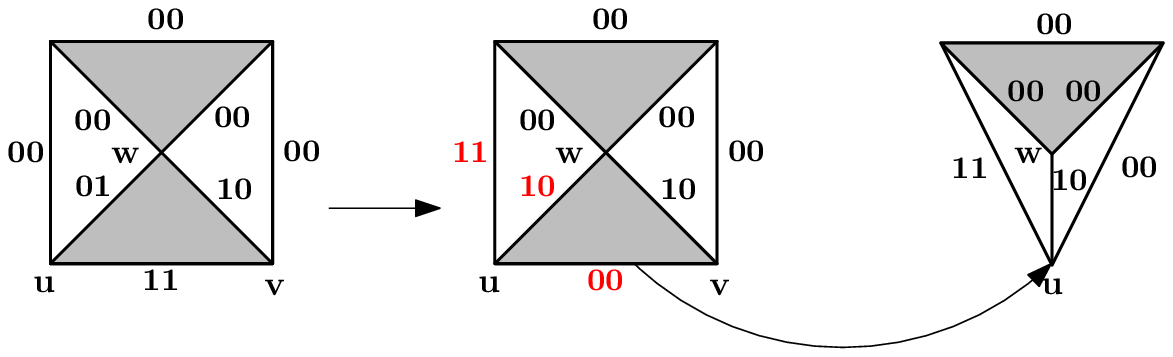}
\end{tabular}
\end{center}
\caption{Annotation updates for elementary collapse: inclusion of a triangle
to satisfy the link condition (upper row), annotation transfer and
actual collapse (lower row); annotation $11$ of the vanishing edge $uv$
is added to all edges (cofaces) adjoining $u$.}
\label{fig:annotation}
\end{figure*}
%It is known that if $K_i$ is a $2$- or $3$-manifold, the 
%underlying spaces $|K_i|$ and $|K_{i+1}|$ are homeomorphic 
%if $(u,v)$ satisfies the link
%condition. Instead of homeomorphism, if we require only
%homotopy equivalence, we can relax
%the condition of $K_i$ being only $2$- or $3$-manifolds,
%thanks to a result in~\cite{ALS11}. 
\begin{proposition}\cite{ALS11}
If an elementary collapse $f_i: K_i\rightarrow K_{i+1}$
satisfies the link condition, then the
underlying spaces $|K_i|$ and $|K_{i+1}|$ remain 
homotopy equivalent and hence the induced
homomorphisms $f_{i_*}: H_*(K_i)\rightarrow
H_*(K_{i+1})$ and $f_i^*: H^*(K_i)\leftarrow H^*(K_{i+1})$
are isomorphisms.
\label{lnk}
\end{proposition} 

\begin{wrapfigure}{l}{1.1in}
\input{commute2.pstex_t}
\end{wrapfigure}
If an elementary collapse satisfies the link condition, we 
can perform the collapse knowing that the (co)homology does not change.
Otherwise, we know that the (co)homology is affected 
by the collapse
and it should be reflected in our updates for annotations.
The diagram at the left provides a precise means to carry out the
change in (co)homology.
Let $S$ be the set of simplices in non-decreasing order of dimensions, 
whose absence from $K_i$ makes $(u,v)$ violate the link condition. 
For each such simplex $\sigma\in S$, we modify the annotations
of every simplex which we would have done if $\sigma$ were to be
inserted. Thereafter, we carry out the rest of the 
elementary collapse. In essence, implicitly,
we obtain an intermediate
complex $\hat{K_i} = K_i \cup S$ 
where the diagram on the left commutes.
Here, $f_i'$ is induced by the same vertex map that induces $f_i$,
and $j$ is an inclusion. 
%with $f_i'$ being the extension of $f_i$ to $\hat{K_i}$.
This means that the persistence of $f_i$ is identical to
that of $f_i'\circ j$ which justifies our action of elementary
inclusions followed by the actual collapses.

We remark that this is the only place where we may insert
implicitly a simplex $\sigma$ in the current approach.
The number of such $\sigma$ is usually much smaller than the number of simplices 
in the cone $u * \cst v$ that we would need to insert for the algorithm 
using coning. %recall Figure \ref{fig:compare}.

Next, we transfer annotations in $\hat{K_i}$. This step
locally changes the annotations for simplices containing the vertices
$u$ and/or $v$. The following definition facilitates the description. 
\begin{definition}
For an elementary collapse $f_i:K_i\rightarrow K_{i+1}$, 
a simplex $\sigma\in K_i$ is called
{\em vanishing} if the cardinality of $f_i(\sigma)$ is one less than that
of $\sigma$. Two simplices $\sigma$ and $\sigma'$ are
called {\em mirror} pairs if one contains $u$ and the other $v$, and
share rest of the vertices. In Figure~\ref{fig:annotation}(lower row), 
the vanishing
simplices are $\{\{u,v\}, \{u,v,w\}\}$ and the mirror pairs are
$\{\{u\},\{v\}\}$, $\{\{u,w\},\{v,w\}\}$.
\end{definition}

In an elementary collapse that sends
$(u,v)$ to $u$, all vanishing simplices need
to be deleted, and all simplices containing $v$ need to be pulled
to the vertex $u$ (which are their mirror partners).
% which requires changing their annotations.
We update the annotations in such a way that the annotations
of all vanishing simplices become zero, and those of each pair of mirror simplices
become the same. Once this is achieved, the collapse is implemented by
simply deleting the vanishing simplices and replacing
$v$ with $u$ in all simplices
containing $v$ without changing their annotations.
The following proposition provides the 
justification behind the specific update operators that we perform.
\begin{proposition}
Let $K$ be a simplicial complex and
$\ann: K(p)\rightarrow {\mathbb Z}_2^g$ be a valid annotation.
Let $\sigma\in K(p)$ be any $p$-simplex and $\tau$
any of its $(p-1)$-faces. Adding $\ann_{\sigma}$ to the annotation
of all cofaces of $\tau$ of codimension $1$
produces a valid annotation for $K(p)$.
Furthermore, the cohomology basis corresponding to the annotations remains
unchanged by this modification.
\label{annot-push}
\end{proposition}
\begin{proof}
Let $\{[\phi_1],\ldots,[\phi_g]\}$ be a cohomology basis
of $H^p(K)$ corresponding
to $\ann: K(p)\rightarrow {\mathbb Z}_2^g$ as stated in
Proposition~\ref{annot-cohom}. Let $T$ be the set of cofaces of
$\tau$ of codimension $1$ and
\begin{eqnarray*}
\phi'_i(\sigma')=
\left\{\begin{array}{ll}
                 \phi_i(\sigma') & \mbox{ if } \sigma'\in K(p)\setminus T\\
                 \phi_i(\sigma') + \phi_i(\sigma) & \mbox{ if } \sigma' \in T
        \end{array}
\right.
\end{eqnarray*}
By construction, $\phi_i'$ is the cochain that corresponds to the
new annotation obtained by adding $\ann_{\sigma}$ to that of the simplices
in $T$. We prove that $\phi_i'$ is a cocycle in the class $[\phi_i]$.
Therefore, $\{[\phi_1'],\ldots,[\phi'_g]\}$ is a cohomology basis
of $H^p(K)$. The new annotation is valid by
Proposition~\ref{annot-cohom} and the cohomology bases remain unchanged.

If $\phi_i(\sigma)=0$, we have $\phi'_i=\phi_i$ and thus $[\phi_i']= [\phi_i]$
trivially. So, assume that $\phi_i(\sigma)=1$. In this case
$\phi'_i=1+\phi_i$ on $T$ and equals $\phi_i$ everywhere else.
Consider the $(p-1)$-cochain $\phi$ defined by $\phi(\tau)=1$
and $\phi(\tau')=0$ for every $\tau'\in K_{p-1}\setminus \tau$.
Then the coboundary $\delta \phi$ is a $p$-cochain that is $1$
for every simplex in $T$ and $0$ on other $p$-simplices.
We can write $\phi_i'=\phi_i + \delta \phi$. It follows that
$[\phi_i']= [\phi_i]$.
\end{proof}
~\\
%%%%%%%%%%%%%%%%%%%%%%%%%figures%%%%%%%%%%%%%%%%%%%%%%%%%%%%%%%%%
%%\begin{wrapfigure}{r}{3.3in}
%\begin{figure}
%\centering
%\includegraphics[scale=0.7]{images/Annotation_Fig.eps}
%%\input{images/annotation.pstex_t}
%\caption{Updating annotations: 1) inserting the blue triangle in upper row; it kills a $1$-cycle whose annotation was $10$ originally. The first bit 1 is added
%to the single edge annotation which is 11 and make it 01. 
%The first bit is deleted from all annotations.
%2) collapsing $uv$ in lower row; the annotation 11 of $uv$
%is added to all edges containing $u$ before the collapse. 
%Only edge annotations are shown.}
%\label{fig:annotation}
%%\end{wrapfigure}
%\end{figure}

Consider an elementary collapse $f_i:K_i\rightarrow K_{i+1}$ that
sends $(u,v)$ to $u$.
We update the annotations in $K_i$ as follows. 
First, note that the vanishing simplices are exactly those simplices
containing the edge $\{u,v\}$.
For every $p$-simplex
containing  $\{u,v\}$, i.e., a vanishing simplex,
exactly two among its $(p-1)$-faces
are mirror simplices, and all other remaining $(p-1)$-faces
are vanishing simplices.
Let $\sigma$ be a vanishing $p$-simplex and $\tau$ be its
$(p-1)$-face that is a mirror simplex containing $u$.
We add $\ann_{\sigma}$ to the 
annotations of all cofaces of $\tau$ of codimension $1$ including $\sigma$. 
We call this an \emph{annotation transfer} for $\sigma$. 
By Proposition \ref{annot-push}, the new annotation generated 
by this process corresponds to the old cohomology basis 
for $K_i$. 
This new annotation has $\ann_{\sigma}$ as zero since
$\ann_{\sigma}+\ann_{\sigma}=0$. See the the lower row 
of Figure~\ref{fig:annotation}.

%Observe that the above process also applies to all faces of $\tau$ if $|\tau|>1$. 
We perform the above operation for each vanishing simplex. 
It turns out that by using the relations of 
vanishing simplices and mirror simplices, 
each mirror simplex eventually acquires an identical annotation  to that of its partner.
Specifically, we have the following observation. 
\begin{proposition}
After all possible annotation transfers involved in a collapse, 
(i) each vanishing simplex has a zero annotation; 
and (ii) each mirror simplex $\tau$ has the same annotation as its mirror partner simplex $\tau'$.
\label{claim:transferresult}
\end{proposition}
\begin{proof}
%It then follows that our algorithm as described above 
Our algorithm performs an annotation transfer for
every vanishing simplex.
Furthermore, the annotation transfer for a
vanishing simplex $\sigma$ does not affect
the annotation of any other vanishing simplex. 
Hence, the annotation of each vanishing
simplex $\sigma$ is updated exactly once after which it becomes
zero and remains so throughout the rest of the annotation
transfers for other vanishing simplices.
This proves claim (i).

For claim (ii), consider a pair of $(p-1)$-dimensional
mirror simplices $\tau = \{ u, u_2, \ldots, u_p \}$ and
$\tau' = \{ v, u_2, \ldots, u_p \}$.
Since $(u, v )$ satisfies the link condition, it is necessary
that the $p$-simplex $\alpha=\{ u, v, u_2, \ldots, u_p \}$ must 
exist in $\hat{K}_i$.
Thus, we have that $\ann_{\partial \alpha} = 0$. On the other hand,
other than $\tau$ and $\tau'$, any $(p-1)$-face of $\alpha$ is a
vanishing simplex, and by Claim (i), in the end, has 
zero annotation.
Therefore, after all annotation transfers, $\ann_{\partial \alpha} = \ann_\tau + \ann_{\tau'} = 0$, implying that $\ann_\tau = \ann_{\tau'}$.
\end{proof}
%\\
~\\

Subsequent to the annotation transfer, the annotation of $\hat{K}_i$
fits for actual collapse since each pair of mirror simplices which
are collapsed to a single simplex get the identical annotation
and the vanishing simplex acquires the zero annotation. 
Furthermore, Proposition~\ref{annot-push} tells us that the cohomology
basis does not change by annotation transfer which aligns with the
fact that $f_{i}'^*: H^*(\hat{K_i})\leftarrow H^*(K_{i+1})$ is indeed an
isomorphism. Accordingly, no time stamp changes after the annotation
transfer and the actual collapse. 
The next section
presents formal statements justifying the algorithm for annotation updates. 

\section{Justification}
%\label{sec:justification}
\label{SEC:JUSTIFICATION}

\vspace*{0.1in}
In this section we justify the algorithm for annotation updates.
Generically assume $f: K\rightarrow K'$ is an elementary map
inducing a homomorphism $H^p(K)\stackrel{f^*}{\leftarrow} H^p(K')$
in ${\cal M}^*$ where $K=K_i$ and
$K'=K_{i+1}$ for some $i\in \{1,\ldots,n\}$.
Let $\{\phi_i'\}$ be the cochains
corresponding to the annotations computed for $K'$ given
a valid annotation $\{\phi_i\}$ for
$K$.
First, we show that
the computed annotations remain valid (proof in Appendix~\ref{appendix:justification}),
that is, $\{\phi_i'\}$ indeed forms a cohomology basis for
$H^p(K')$. Then, we show in Propositions~\ref{element-just}
and \ref{prop:collapse} how
the cohomology bases $\{[\phi_i']\}$ and $\{[\phi_i]\}$
for $K'$ and $K$ respectively correspond
under the homomorphism $f^*$. The time stamps
used by the algorithm concur with this correspondence.
\begin{proposition}
Let $\{[\phi_i]\}$ be the cohomology basis for $H^p(K)$ given by
a valid annotation for $K$ and
$\{\phi'_i\}$ be the cochains corresponding to the annotation computed
for $K'$ by the
update algorithm. Then, $\{[\phi_i']\}$ is a cohomology
basis for $H^p(K')$.
\label{prop:correctbasis}
\end{proposition}
First, we focus on when $f$ is an elementary inclusion.
If $f$ is an elementary inclusion, it is known that
in the persistence module, $f_*$ is either injective
in which case a new class is born,
or surjective in which case a class is killed.
In the dual module with cohomology, $f^*$ switches the role, that is,
$f^*$ is surjective when $f_*$ is injective and vice versa.
\begin{proposition}
Let $\sigma :=K'\setminus K$ be a $p$-simplex inserted for
inclusion $f: K\rightarrow K'$.
\begin{enumerate}
\item[i.] $f_*$ is injective ($f^*$ is surjective):
Let $[\phi_1],\ldots,[\phi_g]$
be a basis of $H^p(K)$ given by a valid annotation.
Let $\phi'_1,\ldots,\phi'_{g+1}$ be the cochains
that correspond to the annotation computed for $K'$ by the update
algorithm. Then, $[\phi_i]=f^*([\phi_i'])$
for $i=1,\ldots, g$ and $f^*([\phi_{g+1}'])=0$.
Cohomology bases for dimensions other than $p$ remain unchanged.

\item[ii.] $f_*$ is surjective ($f^*$ is injective):
Let $\{[\phi_1],\ldots, [\phi_g]\}$
be a basis for $H^{p-1}$ given by a valid annotation.
Let $\phi'_1,\ldots, \phi'_{g-1}$ be the cochains
that correspond to the annotation computed
for $K'$ by the update algorithm which deletes the
$u$th element. Then,
for $1\leq i< u$, $[\phi_i]=f^*([\phi_i'])$ and
for $i\geq u$, $[\phi_{i+1}]=f^*([\phi_i'])$.
Cohomology bases for dimensions other than $p-1$ remain
unchanged.
\end{enumerate}
\label{element-just}
\end{proposition}

\begin{proof}
We provide the proof for (i) here and defer the proof of (ii) to the appendix.
Recall that $\sigma=K'\setminus K$ is a $p$-simplex inserted for inclusion.
We observe that when $f$ is an inclusion, we have
$f_\#(z)=z$ for any cycle $z$ in $K$ where $f_\#$ denotes the
chain map induced by $f$.

Consider the case for (i).
In this case, $\sigma$ creates a new $p$-cycle and no other $k$-cycle
for $k\not = p$. The annotations for $k$-simplices for $k\not= p$ are
not changed. Therefore, a basis for $H^k(K)$ for $k\not= p$
remains so in $H^k(K')$. So, we can focus only on the case $k=p$.
The algorithm updates the annotations of $p$-simplices in $K$
by appending a $0\in \Z_2$ for everyone except the simplex $\sigma$
which gets a $1\in\Z_2$.
The definition
of the homomorphism $H^p(K)\stackrel{f^*}{\leftarrow} H^p(K')$ provides that,
for every $i\in \{1,\ldots, g\}$,
there is a cocycle $\phi$ defined by the homomorphism
$z\mapsto \phi_i'(f_\#(z))$ where
$f^*([\phi_i'])=[\phi]$.
If $[z_1],\ldots, [z_g]$
is a basis of the homology group $H_p(K)$,
the class $[\phi]$ is uniquely determined
by the vector $[\phi(z_1),\ldots,\phi(z_g)]$.
We have
\begin{eqnarray*}
[\phi(z_1),\ldots,\phi(z_g)]
&=&[\phi_i'(f_\#(z_1),\ldots, \phi_i'(f_\#(z_g))]\\
&=&[\phi_i'(z_1),\ldots,\phi_i'(z_g)]\\
&=&[\phi_i(z_1),\ldots,\phi_i(z_g)].
\end{eqnarray*}
The last equality follows from the fact that
$\phi_i'(z_j)=\phi_i(z_j)$ because the $i$th element 
in the annotation for $p$-simplices remains the same
for $1\leq i \leq g$. Since $\phi$ and $\phi_i$ evaluate
the basis $[z_1],\ldots,[z_g]$ the same, we have
$[\phi]=[\phi_i]$, that is, $f^*([\phi_i'])=[\phi_i]$
as we are required to prove. Following the same argument
we see that
$
[\phi_{g+1}'(f_\#(z_1),\ldots, \phi_{g+1}'(f_\#(z_g))]
=[0,\ldots,0]
$
since the cycles $z_i$ for $i=1,\ldots,g$ do not include
the simplex $\sigma$ and thus have the element $0$ in the
$(g+1)$-th position of the annotation for every simplex in them.
Clearly, $f^*([\phi_{g+1}'])=0$.\\

Similar to the case above, one can prove for case (ii) 
(see Appendix~\ref{appendix:justification}) that
$$
[\phi(z_1),\ldots,\phi(z_g)]
=[\phi_i(z_1),\ldots,\phi_i(z_g)]
$$
giving $f^*([\phi_i'])=[\phi]=[\phi_i]$ for $i=\{1,\ldots,u-1\}$.
The case for $i\in \{u+1,\ldots, g\}$ can be proved similarly.
The only caveat is that the $u$th element is zeroed out in annotation,
so there is a left shift of the elements lying to the right 
of the $u$th element 
in the annotation which accounts for the assertion
$f^*([\phi_i'])=[\phi_{i+1}]$.
\end{proof}
~\\

Next, we consider the case when $f=f_{i}$ is an elementary
collapse. Recall that we implement such a collapse as a composition
of elementary inclusions $j$ and a vertex collapse $f'$ where
$f=f'\circ j$. This induces the following sequence
$H^*(K) \stackrel{j^*}{\leftarrow} H^*(\hat{K})
\stackrel{f'^*}{\leftarrow} H^*(K')$.
Since we have already argued about inclusions, we only need to
show that the annotation updates reflect the map
$f'^*$.

\begin{proposition}
Let $[\phi_1],\ldots,[\phi_g]$ be a basis of
$H^p(\hat{K})$ given by a valid annotation.
Let $\phi'_1,\ldots,\phi_g'$ be the cochains
that correspond to the annotation computed
for $K'$ by the update algorithm. Then,
$[\phi_i]=f'^*([\phi_i'])$
for $i=1,\ldots,g$.
\label{prop:collapse}
\end{proposition}
\begin{proof}
First, recall that $f'_*$ and hence
$f'^*$ is an isomorphism due to Proposition~\ref{lnk}
as the vertex pair $(u,v)$ satisfies the link condition in
$\hat{K}$. Let $[z_1],\ldots,[z_g]$ be a basis
in $H_*(\hat{K})$. As before, let $\phi$ be a cocycle
defined by the homomorphism $z\mapsto\phi'_i(f'_\#(z))$
where $f'^*([\phi_i'])=[\phi]$.
We have
$[\phi(z_1),\ldots,\phi(z_g)]
=[\phi_i'(f'_\#(z_1)),\ldots, \phi_i'(f'_\#(z_g))].
$
Recall that we first
carry out an annotation transfer in $\hat{K}$ to match the
annotations for the mirror simplices and to zero out
the annotations for the vanishing simplices.
This update does not change the cohomology classes thanks to
Proposition~\ref{annot-push}. So, we focus on the update due to the
vertex collapse. Every pair of mirror simplices carries their annotation
into the collapsed simplex, and vanishing simplices lose their
zero annotations as they disappear. In effect, we have
$\phi_i'(f'_{\#} (z_j))=\phi_i(z_j)$,
giving us that
$[\phi(z_1),\ldots,\phi(z_g)]
=[\phi_i(z_1),\ldots, \phi_i(z_g)].
$
Therefore, $f'^*([\phi_i'])=[\phi]=[\phi_i]$ for
$i=1,\dots,g$.
\end{proof}

\section{Application to topological data analysis}
\label{SEC:PDAPPROX}
In topological data analysis, several applications and approaches
use Rips complex filtration~\cite{ALS11,DSW11,Sheehy}.
The computation
of the persistence diagram or its approximation for a Rips filtration appears
to be a key step in these applications.
However, the size of this filtration becomes a bottleneck
because of the inclusive nature of Rips complexes.
A natural way to handle this problem is to successively subsample 
the input data and build a filtration on top of them.
We show how one can apply our results from  
previous sections to approximate the
persistence diagrams of a Rips filtration from such a sparser filtration. 
%Let $\Rips^r (V)$ denote the Rips complex on point set $V$ with parameter $r$. 
Given a set of points $V\subset \mathbb{R}^d$ 
(Similar to \cite{Sheehy}, results in this section can be extended to 
any metric space with doubling dimension $d$.), 
let $\Rips^r (V)$ denote the Rips complex on the point set $V$ with parameter $r$. That is,  
a $k$-simplex $\sigma=\{u_0, \ldots, u_k \} \subseteq V$ is in $\Rips^r(V)$ if and only if $\| u_i - u_j \| \le r$ for any $i, j \in [0,k]$. 
We present an algorithm to approximate the persistence diagram for the following Rips filtration. The parameters $\alpha>0$ and $0\leq \eps \leq 1$ are assumed to be preselected. 
\begin{equation}
%\hookrightarrow \Rips^{(1+\eps)^2 \alpha}(V)
\Rips^\alpha(V) \hookrightarrow \Rips^{(1+\eps)\alpha} (V)  \cdots \hookrightarrow \Rips^{(1+\eps)^m \alpha}(V). 
\label{eqn:ripsfiltration}
\end{equation}

The number of $k$-simplices in a Rips complex 
with $n$ vertices can be $\Theta(n^{k+1})$. 
This makes computing the persistent homology of the above filtration costly. 
In \cite{Sheehy}, Sheehy proposed to approximate the persistence diagram 
of the above filtration by another Rips filtration where 
each simplicial complex involved has size only linear in $n$. 
This approach allows vertices to be collapsed (deleted) 
with a weighting scheme when the parameter 
$r$ for the Rips complex becomes large, which helps to keep the size of 
the simplicial complex at each stage small. 
%He proposed a weighting scheme that helps approximating
%the original persistent module 
%by another arising out of a filtration connected by inclusions and 
%thereby allowing classical algorithms such as \cite{ELZ02} to be 
%employed. 

In this section, we provide an alternative approach to approximate the 
persistence diagram of the filtration given in (\ref{eqn:ripsfiltration}). 
We achieve sparsification by
subsampling as in~\cite{Sheehy}, 
but our persistence algorithm for
simplicial maps allows us to handle the sequence of complexes 
induced by the clustering / collapsing of vertices directly
instead of an additional weighting scheme. 
We consider two sparsification schemes, one produces a sequence of sparsified 
Rips complexes, and the other produces a sequence of 
graph induced complexes (GICs)
which have been shown to be even sparser in practice~\cite{DFW13}. 
%However, the construction of GIC sequence may be costlier than
%that of the Rips creating a tradeoff of space versus pre-processing time in practice.
Asymptotically, both sequences have sizes linear in number of vertices. 
%Our alternative approach has the same asymptotic time / space complexity as Sheehy's, but is simpler and 
%conceptually cleaner.  
%avoiding the need to design appropriate weighting functions, and we believe it is concepturally cleaner. 

%This new filtration in essence is simply a filtration connected by simplicial maps. 
%Sheehy then developed an elegant weighting scheme to replace the simplicial  maps induced by vertex collapses with a filtration again connected by inclusions, so that classical algorithms such as \cite{ELZ02} can be employed. 
%
%In this section, we provide an alternative approach to approximate the persistence diagram of the filtration given in (\ref{eqn:ripsfiltration}). 
%We follow the same vertex-collapse framework as Sheehy, but handle the filtration connected by simplicial maps directly, and thus avoid the need of re-weighting. 

%we will show that we can convert the filtration in Eqn (\ref{eqn:ripsfiltration}) to a simplicial filtration where each simplicial complex involved has size linear in $n$. 
%The persistence diagram of this simplicial filtration approximates that of Eqn (\ref{eqn:ripsfiltration}). 
%We can then apply the algorithm described in Section \ref{sec:onedirection} to compute the persistence diagram of the simplicial filtration. 

\subsection{Persistence diagram approximation by sparsified Rips complex}
\label{subsec:sparsifyRips}
Given a set of points $V$, we say that $V' \subseteq V$ is a \emph{$\delta$-net} of $V$ if (i) for any point $v \in V$, there exists a point $v' \in V$ such that $\| v - v' \| \le \delta$; and (ii) no two points in $V'$ are within 
$\delta$ distance. 
A $\delta$-net $V'$ can be easily constructed by a standard 
greedy approach by taking furthest points iteratively 
%(see e.g; \cite{Cla06}), 
or by more sophisticated and efficient methods as in~\cite{Cla06,HM06}). 

Now set $V_0 := V$. 
We first construct a sequence of point sets $V_k$, $k = 0, 1, \ldots, m$, such that $V_{k+1}$ is a $\frac{\alpha \eps^2}{2} (1+\eps)^{k-1}$-net of $V_{k}$. 
%where $\lambda < \frac{1}{2}$ and $\eps \geq \frac{2\lambda}{1-2\lambda}$. 
Consider the following vertex maps $\vmap_k: V_{k} \rightarrow V_{k+1}$, for $k \in [0, m-1]$, where for any $v \in V_k$, $\vmap_{k}(v)$ is the vertex in $V_{k+1}$ that is closest to $v$. 
Define $\hv_{k}: V_0 \rightarrow V_{k+1}$ as $\hv_k(v) = \vmap_k \circ \cdots \vmap_0(v)$. 

Each vertex map $\vmap_k$ induces a well-defined simplicial map $\hmap_k: \Rips^{\alpha(1+\eps)^k} (V_k) \rightarrow \Rips^{\alpha(1+\eps)^{k+1}}(V_{k+1})$.  Indeed, 
since $V_{k+1}$ is a $\frac{1}{2} \alpha \eps^2 (1+\eps)^{k-1}$-net of $V_{k}$, 
for each edge $e = \{u,v\}$ from $\Rips^{\alpha(1+\eps)^k}(V_k)$, we have 
\begin{eqnarray*} 
\| \vmap_k(u) - \vmap_k(v) \| &\le& \| u - v \| + \| u - \vmap(u) \| + \|v - \vmap(v) \|\\
& \le &\alpha (1+\eps)^k + \alpha \eps^2 (1+\eps)^{k-1} \\
& \leq &\alpha(1+\eps)^{k+1}.
\end{eqnarray*}
Hence $\vmap_k(u) \vmap_k(v)$ is an edge 
in $\Rips^{\alpha(1+\eps)^{k+1}}(V_{k+1})$. 
Since in a Rips complex, higher dimensional simplices are determined by the edges, every simplex $\{u_0, \ldots, u_d\}$ in $\Rips^{\alpha(1+\eps)^k}(V_k)$ has a well-defined image $\{ \vmap_k(u_0), \ldots, \vmap_k(u_d)\}$ in $\Rips^{\alpha(1+\eps)^{k+1}}(V_{k+1})$. 
Hence, each $\hmap_k$ is well-defined providing the filtration:
%These simplicial maps connect the following sequence of simplicial complexes: 
\begin{equation}
\Rips^\alpha(V_0) \stackrel{\hmap_0}{\longrightarrow} \Rips^{\alpha(1+\eps)}(V_1) \cdots \stackrel{\hmap_{m-1}}{\longrightarrow} \Rips^{\alpha(1+\eps)^m} (V_m). 
\label{eqn:simpfiltration}
\end{equation}

In other words, as the parameter $r = \alpha(1+\eps)^k$ increases, we can simply consider the Rips complex built upon the sparsified data points $V_k$. 
Note that the sequence above is not connected by inclusion maps and thus classical persistent algorithms cannot be applied directly;
%. In \cite{Sheehy}, 
%a different sequence connected by inclusion maps is obtained using
%a weighting scheme and a conglomeration technique; 
while our algorithm from Section \ref{SEC:ALGORITHM} can be used 
here in a straightforward manner. 

Our main observation is that the persistence diagram of the 
sequence of simplicial maps in  (\ref{eqn:simpfiltration}) approximates 
that of the inclusion maps in (\ref{eqn:ripsfiltration}). 
In particular, we show that the persistence modules induced by these two sequences interleave in the sense described in~\cite{CCGGO09}. 

First, we need maps to connect these two sequences. 
For this, we observe that the vertex map $\hv_k: V_0 \rightarrow V_{k+1}$ 
also induces a simplicial map 
$\hh_k: \Rips^{\alpha(1+\eps)^k} (V_0) \rightarrow \Rips^{\alpha(1+\eps)^{k+1}} (V_{k+1})$. 
To establish that this simplicial map is well-defined, it
can be shown that if there is an edge $\{u,v\}$ in $\Rips^{\alpha(1+\eps)^k} (V_0)$, then there is an edge $\hv_k(u) \hv_k(v)$ in $\Rips^{\alpha(1+\eps)^{k+1}} (V_{k+1})$. 
%Indeed, setting $\hv_{-1}(u) = u$ and $\hv_{-1}(v) = v$, we have: 
%\begin{eqnarray*}
%\| \hv_k(u) - \hv_k(v) \| 
%&\le & \| \hv_k(u) - u \| + \| \hv_k(v) - v \| + \| u - v \| \\
%&\le & \sum_{i=0}^k \| \hv_i(u) - \hv_{i-1}(u) \| 
% + \sum_{i=0}^k \| \hv_i(v) - \hv_{i-1}(v) \| + \| u - v \| \\
%&\le& \frac{\alpha\eps^2}{1+\eps} \sum_{i=0}^k (1+\eps)^i + \alpha (1+\eps)^k\\
%& \le& \alpha (1+\eps)^{k+1}. \\
%\end{eqnarray*}
%Hence $\hh_k$ is well-defined. 

\begin{claim}
Each triangle in the following diagram commutes at the homology level.  
\[
\xymatrix @R=1.5pc @C=1.5pc
{
\Rips^{\alpha(1+\eps)^{k}}(V_{0}) \ \ar@{^{(}->}[rr]^-{i_k} \ar[rrd]^{\hh_{k}}
&  & \  \Rips^{\alpha(1+\eps)^{k+1}}(V_{0}) \\
\Rips^{\alpha(1+\eps)^{k}}(V_{k}) \ \ar@{^{(}->}[u]^{j_k} \  \ar[rr]^{\hmap_k}  
&  & \  \Rips^{\alpha(1+\eps)^{k+1}}(V_{k+1}) \ \ar@{^{(}->}[u]^{j_{k+1}} \   
}
\]
Here, the maps $i_k$s and $j_k$s are canonical inclusions. The simplicial maps $\hh_k$ and $\hmap_k$ are induced by the vertex maps $\hv_k: V_0 \rightarrow V_{k+1}$ and $\vmap_k: V_k \rightarrow V_{k+1}$, respectively, as described before. 
\label{claim:trianglescommute}
\end{claim}
\begin{proof}
First, we consider the bottom triangle. 
Note that the vertex map $\hv_k$ restricted on the set of vertices $V_k$ is the same as the vertex map $\vmap_k$. (That is, for a vertex $u \in V_k \subseteq V_0$, $\hmap_k(u) = \hh_k(u)$.) Thus $\hmap_k = \hh_k \circ j_k$. Hence the bottom triangle commutes both at the simplicial complex level and at the homology level. 

Consider the top triangle. 
We claim that the map $j_{k+1} \circ \hh_{k}$ is contiguous to the inclusion $i_k$.
Since two contiguous maps induce the same homomorphisms at the homology level, 
the top triangle commutes at the homology level. 

This claim can be verified by the definition of contiguous maps.
Given a simplex $\sigma \in \Rips^{\alpha(1+\eps)^{k}}(V_0)$, we wish to show that vertices from $\sigma \cup \hh_k(\sigma)$ span a simplex in $\Rips^{\alpha(1+\eps)^{k+1}} (V_0)$. 
Since $\Rips^{\alpha(1+\eps)^{k+1}} (V_0)$ is a Rips complex, we only need to show that for any two vertices $u$ and $v$ from $\sigma \cup \widehat{h}_k(\sigma)$, the edge $uv$ has length less than $\alpha(1+\eps)^{k+1}$ (and thus in $\Rips^{\alpha(1+\eps)^{k+1}} (V_0)$). 
If $u$ and $v$ are both from $\sigma$ or both from $\hh_k(\sigma)$,
then obviously $\| u - v \| \leq \alpha(1+\eps)^{k+1}$. 
Otherwise, assume without loss of generality that $v \in \sigma$ and $u \in \hh_{k}(\sigma)$ where $u=\hv_k(\bar{u})$ for some $\bar{u} \in \sigma$.
It then follows that, 
\begin{eqnarray*}
\|u - v \| &\leq& \|u - \bar{u} \| + \| \bar{u} - v \| \\
&\leq & 
\frac{\alpha \eps^2}{2 (1+\eps)} \sum_{i=0}^{k} (1+\eps)^{i} + \alpha (1+\eps)^{k}
 <  \alpha (1+\eps)^{k+1}.
\end{eqnarray*} 
Therefore, the vertices of $\sigma \cup \widehat{h}_k(\sigma)$ span a 
simplex in $\Rips^{\alpha(1+\eps)^{k+1}}(V_{0})$.
\end{proof}
~\\

The above claim implies that the persistence modules induced by 
sequences (\ref{eqn:ripsfiltration}) and (\ref{eqn:simpfiltration}) are weakly $\frac{\log (1+\eps)}{2}$-interleaved at the log-scale. 
By Theorem 4.3 of \cite{CCGGO09}, we thus conclude with the following: 
\begin{proposition}
The persistence diagram of the sequence (\ref{eqn:simpfiltration}) provides a $\frac{3 \log (1+\eps)}{2}$-approximation of the persistence diagram 
of the sequence (\ref{eqn:ripsfiltration}) at the log-scale. 
\end{proposition}

Finally, since $V_{k+1}$ is a $\delta$-net of $V_k$ for $\delta = \frac{\alpha \eps^2}{2} (1+\eps)^{k-1}$, we can show by a standard packing argument that each 
$\Rips^{\alpha(1+\eps)^k} (V_k)$ is of size linear in $n$. See
Proposition C.1 in appendix. 
%\begin{lemma}
%Suppose the set of input points $V$ are from a metric space with doubling dimension $d$. For $V_k$s constructed as described above, the number of $p$-simplices in $\Rips^{\alpha(1+\eps)^{k+1}} (V_{k+1})$ is $O((\frac{1}{\eps})^{O(dp)} \cdot n)$ for $0\le \eps \le 1$. 
%\label{lem:size}
%\end{lemma}
%
Note that the persistence diagram of the simplicial maps in (\ref{eqn:simpfiltration}) can be computed by our algorithm in Section \ref{SEC:ALGORITHM}. 
Putting everything together, we have the following result. 

\begin{theorem}
Given a set of $n$ points $V$ in a metric space with doubling-dimension $d$ and $0\le \eps \le 1$, we can $\frac{3\log (1+\eps)}{2}$-approximate the persistence diagram of the Rips complex filtration (\ref{eqn:ripsfiltration}) by that of the filtration (\ref{eqn:simpfiltration}). The $p$-skeleton of each simplicial complex involved in (\ref{eqn:simpfiltration}) has size $O((\frac{1}{\eps})^{O(dp)} n)$. 

\end{theorem}

%%%%%%%%%%%%%%%%%%%%%%%%%%%%%%%%%%%%%%%%%%%%%%
\subsection{Persistence diagram approximation by graph induced complex}
\label{subsec:GIC}
We now present an alternative way to construct a sequence of complexes 
for gradually sparsified or subsampled points. 
The \textit{graph induced complex} (GIC) proposed in~\cite{DFW13}
works on a subsample as the sparsified Rips complex does. However, it
contains much fewer simplices in practice. 
In~\cite{DFW13}, it was shown how GIC can be used to estimate the homology 
of compact sets by investigating the persistence of a single simplicial map. 
Here we show how one can build a sequence of GICs to approximate
the persistence diagram of a Rips filtration.
Similar to the case of a sequence of sparsified Rips complexes, simplicial maps occur naturally to connect these GICs in the sequence. 

\begin{definition} Let $G(V)$ be a graph with the
vertex set $V$ and let $\nu: V\rightarrow V'$ be a vertex
map where $\nu(V)=V'\subseteq V$ 
is a subset of vertices.
The graph induced complex $\G(V,V',\nu):=\G(G(V),V',\nu)$ is defined
as the simplicial complex where a $k$-simplex
$\sigma=\{v_1',v_2',\ldots,v_{k+1}'\}$
is in $\G(V,V',\nu)$ if and only if there exists a $(k+1)$-clique
$\{v_1,v_2,\ldots,v_{k+1}\}\subseteq V$ so that $\nu(v_i)=v_i'$ for
each $i\in \{1,2,\ldots,k+1\}$.
To see that it is indeed a simplicial complex,
observe that a subset of a clique is also a clique.
Let $G(V)$ be called the base-graph for $\G(V,V',\nu)$. 
\end{definition}
Intuitively, the vertex map $\nu$ maps a cluster of vertices from $V$ 
to a single vertex $v' \in V'$, and these vertices constitute the 
``Voronoi cell" of the site $v'$. The GIC $\G(V,V',\nu)$ is somewhat 
the combinatorial dual of such a Voronoi diagram.
In our case the base graph $G(V)$ is the $1$-skeleton of the Rips complex
$\Rips^{r}(V_0)$ and the vertex map $\nu$ is the
map $\hat{\pi}_k: V_0\rightarrow V_{k+1}$ as defined
in the last section. 
Denote $\G^{r}(V_0,V_k):=\G(V_0,V_k,\hat{\pi}_{k-1})$ constructed using the 1-skeleton of $\Rips^r(V_0)$ as the base-graph.
%We show that there are simplicial maps $f_k:
%\G^\alpha(V_0, {V}_k) \longrightarrow \G^{\alpha(1+\eps)}(V_0, {V}_k)$
It is easy to show that by the definition of $\hat{\pi}_{k}$ and construction of $V_k$s, the vertex map 
$\vmap_k: V_k \rightarrow V_{k+1}$
 induces a well-defined simplicial map 
$f_k: \G^{\alpha(1+\eps)^{k-1}}(V_0, {V}_k) \longrightarrow \G^{\alpha(1+\eps)^k}(V_0, {V}_{k+1})$, 
giving rise to the following sequence: 
\begin{eqnarray}
\begin{split}
\G^\alpha(V_0, {V}_1) \stackrel{f_1}{\longrightarrow} \G^{\alpha(1+\eps)}(V_0, {V}_2) \stackrel{f_2}{\longrightarrow} 
\G^{\alpha(1+\eps)^{2}}(V_0, {V}_3) \cdots \\
\stackrel{f_{m-1}}{\longrightarrow} \G^{\alpha(1+\eps)^{m-1}} (V_0, {V}_{m}). 
\end{split}
\label{eqn:gicFiltration}
\end{eqnarray}
We prove that the persistence diagram of the above filtration induced 
by simplicial maps $f_k$'s has the same approximation 
factor to the persistence of diagram of the filtration (\ref{eqn:ripsfiltration}) as that of the filtration (\ref{eqn:simpfiltration}).
Thus, we have:
\begin{theorem}
Given a set of $n$ points $V$ in a metric space with doubling-dimension $d$ and $0\le \eps \le 1$, we can $\frac{3\log (1+\eps)}{2}$-approximate the persistence diagram of the Rips complex filtration (\ref{eqn:ripsfiltration}) by that of the filtration (\ref{eqn:gicFiltration}). The $p$-skeleton of each simplicial complex involved in (\ref{eqn:gicFiltration}) has size $O((\frac{1}{\eps})^{O(dp)} n)$. 
\label{thm:GIC}
\end{theorem}
\begin{proof}
In sequence (\ref{eqn:gicFiltration}),
${V}_{k+1}$ is a $\delta_{k+1}$-net of $V_k$ for $\delta_{k+1} = \frac{\alpha \eps^2}{2} (1+\eps)^{k-1}$ ($k=0,1, \ldots, m-1$) as
in the sequence of (\ref{eqn:simpfiltration}).
Now consider $\hv_k: V_0 \rightarrow V_{k+1}$.
It is immediate that $|p\hv_{k}(p)| \leq \frac{\alpha \eps^2}{2 (1+\eps)} \sum_{i=0}^{k}
(1+\eps)^{i} \leq \frac{\alpha \eps}{2} (1+\eps)^{k}$ for each $p \in V_0$.
In other words, $V_{k+1}$ is a $\frac{\alpha \eps}{2} (1+\eps)^{k}$-sample of $V_0$.
%By definition, the graph $G(V_0)$ for $\G^{\alpha(1+\eps)^{k}}(V_0, {V}_{k+1})$ is the $1$-skeleton of $\Rips^{\alpha(1+\eps)^{k}}(V_0)$.
Recall the GIC $\G^{\alpha(1+\eps)^{k}}(V_0, {V}_{k+1})$ is constructed based on the $1$-skeleton of $\Rips^{\alpha(1+\eps)^{k}}(V_0)$ (used as the base-graph). 
It is easy to show that $\hv_{k}(p)$ induces a simplicial map $\hat{f}_{k} : \Rips^{\alpha(1+\eps)^{k}}(V_0) \rightarrow \G^{\alpha(1+\eps)^{k}}(V_0, {V}_{k+1})$.
To prove that the persistence diagram of the sequence (\ref{eqn:ripsfiltration})
is approximated by that of the sequence (\ref{eqn:gicFiltration}),
it is sufficient to show that the sequence (\ref{eqn:ripsfiltration}) interleaves with
the sequence (\ref{eqn:gicFiltration}).
The following claim
%s similar to Claim \ref{claim:trianglescommute}
reveals the desired interleaving property. Its proof is similar to that
of the Claim~\ref{claim:trianglescommute}.
\begin{claim}
Each triangle in the following diagram commutes at the homology level.
\[
\xymatrix @R=1.5pc @C=1.5pc
{
\Rips^{\alpha(1+\eps)^{k}}(V_{0}) \ \ar@{^{(}->}[rr]^-{i_k} \ar[rrd]^{\hat{f}_{k}}
&  & \  \Rips^{\alpha(1+\eps)^{k+1}}(V_{0}) \\
\G^{\alpha(1+\eps)^{k-1}}(V_0, {V}_{k}) \ \ar@{^{(}->}[u]^{j_k} \  \ar[rr]^{f_k} 
&  & \  \G^{\alpha(1+\eps)^{k}}(V_0, {V}_{k+1}) \ \ar@{^{(}->}[u]^{j_{k+1}} \
}
\]
Here, the maps $i_k$s and $j_k$s are canonical inclusions.
The simplicial map $\hat{f}_k$ is induced by the vertex map $\hat{\pi}_{k} : V_0 \rightarrow {V}_{k+1}$, and the simplicial map $f_k = \hat{f}_k \circ j_k$.
\label{claim:giccommute}
\end{claim}
\end{proof}
~\\

Note that for every edge $uv$ in $\G^{\alpha(1+\eps)^{k}}(V_0, {V}_{k+1})$,
there is an edge $ab$ in $\Rips^\alpha(1+\eps)^{k}(V_0)$ such that $\hat{\pi}_{k}(a) = u$ and  $\hat{\pi}_{k}(b) =v$.
Since ${V}_{k+1}$ is a  $\frac{\alpha \eps}{2} (1+\eps)^{k}$-sample of $V_0$, one has that
$$|uv| \leq |ua| + |vb| + |ab| \leq \alpha \eps (1+\eps)^{k} + \alpha(1+\eps)^{k} = \alpha(1+\eps)^{k+1}.$$
Therefore, the 1-skeleton of the graph induced complex 
$\G^{\alpha(1+\eps)^{k}}(V_0, {V}_{k+1})$ is a
subcomplex of the $1$-skeleton of 
$\Rips^{\alpha(1+\eps)^{k+1}}({V}_{k+1})$.
Consequently, $\G^{\alpha(1+\eps)^{k}}(V_0, {V}_{k+1})$ is a subcomplex of
$\Rips^{\alpha(1+\eps)^{k+1}}({V}_{k+1})$
which is the maximal simplicial complex containing its $1$-skeleton.
This observation implies that the sequence (\ref{eqn:gicFiltration}) has
smaller size compared to
the sequence (\ref{eqn:simpfiltration}).
Furthermore, although the asymptotic space complexity of each GIC is the 
same as that of the sparsified Rips complex, in practice, the size of 
GICs can be far smaller; see \cite{DFW13}. 
However, the construction of each GIC is more expensive, as one needs to compute each $\G^{\alpha(1+\eps)^{k}} (V_0, {V}_{k+1})$ from the Rips complex $\Rips^{\alpha(1+\eps)^{k}}(V_0)$ built on the original vertex set $V_0$, instead of 
the vertex set from the previous complex 
$\G^{\alpha(1+\eps)^{k-1}} (V_0, {V}_{k})$.
Hence there is a trade-off of space versus time 
for the approaches given in 
Section \ref{subsec:sparsifyRips}.
%and Section~\ref{subsec:GIC}.

%%%%%%%%%%%%%%%%%%%%%%%%%%%%%%%%%%%%%%%%%%%%%
\section{Conclusions}
In this paper, we studied algorithms to compute the persistence diagram of a (monotone) filtration connected by simplicial maps efficiently. 
%Various versions of the standard persistence algorithms have been previously proposed for filtrations induced by inclusion maps (see \cite{ELZ02,DMV11,DMV11b} for computing non-zigzag persistence, and \cite{CSM09} for computing the zigzag persistence). 
As discussed in \cite{DMV11b}, 
the algorithm based on the cohomology view 
%(with some optimizations in the maintenance of a basis) 
in \cite{DMV11} has a good practical performance for
the case of computing inclusion-induced non-zigzag persistence.
Our annotation-based algorithm extends such a view of maintaining
an appropriate cohomology basis to the case of vertex collapses.
This allows us to compute the persistence diagram for a filtration 
connected by simplicial maps directly and efficiently. 
The coning approach in Section~\ref{sec:simulate} works for any finite
fields though the collapse based algorithm in Section~\ref{SEC:ALGORITHM}
currently works 
with $\mathbb{Z}_2$ coefficients only.
Although inclusions can be handled under other finite field
coefficients, it is not clear how to handle collapses efficiently. 
%As an application we show that such a filtration
%does arise in 
%approximating the persistence diagram of 
%a Rips complex filtration that also incorporates vertex collapses. 

%Currently, the improved algorithm based on annotations works only for 
%filtrations induced by a sequence simplicial maps in
%a monotone direction. It would be interesting to see whether 
%this approach can be extended to handle efficiently a zigzag sequence of 
%simplicial maps as well. 

We believe that, as the scope of topological data analysis 
continues to broaden, further applications based on simplicial maps will arise. 
Currently, an efficient implementation of the persistence algorithm taking
advantage of the compressed representation of
annotations has been reported in~\cite{BDM13}. 
We have also developed an efficient implementation of the persistence
algorithms for simplicial maps in the same vein. The software
named {\sf SimpPers} is available from authors' web-pages.

%\end{document}  % This is where a 'short' article might terminate

%ACKNOWLEDGMENTS are optional
%\section{Acknowledgments}
%This section is optional; 

\section*{Acknowledgment} We acknowledge the helpful comments of the
reviewers and the support of the NSF grants
CCF 1116258, CCF 1064416, CCF 1318595, and CCF 1319406.

\appendix
\section{Missing proof from Section \ref{SEC:SIMPLICIALMAPS}} 
\label{appendix:A}

\textbf{Proof of Proposition~\ref{elementary}.}
Consider the surjective simplicial map $f': K \rightarrow f(K)$ defined as $f'(\sigma) = f(\sigma)$ for any $\sigma \in K$.
Writing $V=V(K)$ and $V'=V(K')$, we have $f'_V=f_V$.
%induced by $\fv$; that is, for any $\sigma \in K$ spanned by $u_0, \ldots, u_d$, $f'(\sigma)$ is spanned $\fv(u_0), \ldots, \fv(u_d)$.
The simplicial map $f: K \rightarrow K'$ is a composition
$i \circ f'$, where $i: f(K) \hookrightarrow K'$ is the canonical inclusion $f(K)\subseteq K'$.
Obviously, the inclusion $i$ can be easily decomposed into a sequence of elementary inclusions.
We now show that $f'$ can be decomposed into a sequence of elementary collapses. 

Let  $A:= \{ v \in V' \mid | \fv^{-1}(v) | > 1 \}$.
Hence $\fv$ maps injectively onto $V' \setminus A$.
Order vertices in $A$ arbitrarily as $\{ v_1, \ldots, v_k \}$, $k = |A|$, and let $A_i$ denote $\fv^{-1} (v_i)$.
We now define $f_i$ and $K_i$ in increasing order of $i$.
For the base case, set $K_0 = K$.
For any $i > 0$, consider the vertex map $f_{\mathrm{A}_i}$ which is the injective map on $V(K_{i-1}) \setminus A_i$, but maps $A_i$ to $v_i$.
We set $f_i$ to be the simplicial map induced by this vertex map $f_{\mathrm{A}_i}$, and set $K_i := f_i(K_{i-1})$. By construction, $f_i$ is a 
surjective simplicial map.

It is easy to see that the vertex map $f_{\mathrm{A}_k} \circ \cdots \circ f_{\mathrm{A}_1}$ equals $\fv$. Hence, the induced simplicial map $f_k \circ \cdots \circ f_1 : K \rightarrow K_k$ equals $f': K \rightarrow f(K)$.
% thus satisfies that $f'(\sigma) = f(\sigma)$ for any $\sigma \in K$, and $K_k$ obtained this way equals $f(K)$.
Furthermore, each $f_i$ can be decomposed into a sequence of elementary collapses, each induced by a vertex map that maps only two vertices from $A_i$ to $v_i$.

\begin{proposition}
The simplicial maps $i'\circ f$ and $i$ are contiguous.
\label{contig}
\end{proposition}
\begin{proof}
By definition of contiguous maps, we need to show that for any simplex $\sigma \in K$, $i(\sigma) \cup i' \circ f(\sigma)$ is a simplex in $\hat{K}$. Note that $i(\sigma) = \sigma$.

First assume that the simplex $\sigma$ is not in $\st v$.
Since $f$ is an elementary collapse, we have $f(\sigma) = \sigma$ and $i' \circ f (\sigma) = \sigma$. Hence, $i(\sigma) \cup (i' \circ f)(\sigma)$ equals $\sigma$ which is also a simplex in $\hat{K}$.
%which is the same as $i(\sigma)$ in $\hat{K}$.

Now assume that $\sigma \in \st v$, and
$\sigma=\{u_0, \ldots, u_d \} \cup \{v\}$.
Since $f(v) = \newfixpt$, $f(\sigma)=\{u_0, \ldots, u_d \} \cup \{ \newfixpt \}$, and so is $(i' \circ f)(\sigma)$.
Hence the union of  $i(\sigma)$ and $(i' \circ f)(\sigma)$ is
$ B:= \{u_0, \ldots, u_d \} \cup  \{ u,\newfixpt \}$,
which is the simplex $\sigma\cup\{u\}$.
On the other hand, by construction of $\hat{K}$,
the simplex $\sigma \cup \{\newfixpt\}$ is necessarily in $\hat{K}$.
Hence $i(\sigma) \cup (i' \circ f)(\sigma)$ is a simplex in $\hat{K}$ in this case too. Hence the maps $i$ and $i' \circ f$ are contiguous.
\end{proof}

\begin{proposition}
$i_*': H_*(K')\rightarrow H_*(\hat{K})$ is an isomorphism.
\label{inclu-iso}
\end{proposition}
\begin{proof}
Consider the projection map $\pi: \hat{K}\rightarrow K'$ induced
by the vertex map
\begin{eqnarray*}
\pi(p)& =~~  \begin{cases}
\newfixpt &\mbox{ if $p =v$}\\
  p &\mbox{ otherwise.}
\end{cases}
\end{eqnarray*}
Let $\id_{\hat{K}}$ denote the identity map on $\hat{K}$.
We show that: (i) $\pi$ is an elementary collapse, and (ii) the composition $i' \circ \pi$ and $\id_{\hat{K}}$ are contiguous.
%It then follows from Lemma 3.3 of \cite{Sheehy} that $i'_*$ is an isomorphism.
It is easy to see that $\pi\circ i'$ is $\id_{K'}$. Then, $i'$ is a
(simplicial) homotopy equivalence and hence $i'_*$ is an isomorphism.

Specifically, consider an arbitrary simplex $\sigma \in \hat{K}$.
Let $\hat{X} = \st\{u, v\}$ be the star of $\{u,v\}$ in $\hat{K}$.
If $\sigma \notin \hat{X}$, then by the construction of $\hat{K}$, $\sigma \in K'$. In other words, $\pi(\sigma)=\sigma$ indeed exists in $K'$ in this case. Furthermore, $i' \circ \pi(\sigma) = \sigma$ and thus $(i' \circ \pi)(\sigma) \cup \id_{\hat K}(\sigma) = \sigma \in \hat{K}$.

Now consider the case $\sigma \in \hat{X}$, and assume that $\sigma=\{ u_0, \ldots, u_d \} \cup A$ with $A \subseteq \{ u,v\}$.
To show that $\pi$ is well defined, we need to show that
$\pi(\sigma)=\{ u_0, \ldots, u_d, \newfixpt \}$, is indeed a simplex in $K'$.

(i) If $\newfixpt \notin A$, then by the construction of $\hat{K}$, $\sigma$ has a pre-image in $K$ under the inclusion $i: K \rightarrow \hat{K}$. Hence $\sigma$ must also be a simplex in $K$, and under the map $f$, it is mapped to the
simplex $\{ u_0, \ldots, u_d, \newfixpt \}$ in $K'$. As such, $\pi(\sigma)$ exists in $K'$ in this case.
(ii) If $\newfixpt \in A$, then the simplex $\sigma''=\{u_0, \ldots, u_d \} \cup (A \setminus \{\newfixpt \})$ must exist in the closed star
$\overline{{\st}_{K} \{u, v\}}$ of $\{u,v\}$ in $K$.
%The simplex $\sigma''$ is in $K$ and
Hence $K$ contains a simplex $\sigma'' \cup \{x\}$ with $x$ being
either $u$ or $v$.
Under the map $f$, the image of $\sigma'' \cup \{x\}$ in $K'$
is $\{ u_0, \ldots, u_d, \newfixpt \}$,
hence $\pi(\sigma)$ is well-defined in $K'$ in this case too.

Furthermore, in both (i) and (ii) above, $(i' \circ \pi)(\sigma) = \pi(\sigma)$, and $\pi(\sigma)$ is a face of the simplex $\sigma$. Hence
$(i' \circ \pi)(\sigma) \cup \id_{\hat K}(\sigma) = \sigma \in \hat{K}$.
Putting everything together, we have that $i' \circ \pi$ and the identity map $\id_{\hat{K}}$ are contiguous.
%\begin{eqnarray*}
%\pi(p)& = & \fixpt \mbox{ if $p=v$}\\
% & = & p \mbox{ otherwise.}
%\end{eqnarray*}
%The map $\pi$ is a homotopy equivalence. Its homotopy inverse is
%the inclusion $i'$. Therefore, $i'$ is also a homotopy equivalence.
%It follows that $i'_*$ is an isomorphism.
\end{proof}

\section{Missing Details from Section~\ref{SEC:JUSTIFICATION}}
\label{appendix:justification}

\paragraph{Proof of Proposition~\ref{prop:correctbasis}}
Let $\ann^i_z$ denote the annotation of a cycle $z$ in $K_i$.

Case (i) of elementary inclusion: For $k\neq p$, any $k$-cycle in $K_{i+1}$ was
a $k$-cycle in $K_i$ and the annotations for $k$-simplices are not
altered for $k\neq p$. So, a valid annotation of $K_i$ for $k\neq p$
remains so after inclusion of a $p$-simplex $\sigma$. Now consider
two $p$-cycles $z$ and $z'$ in $K_{i+1}$. We need to show that
$\ann^{i+1}_{z}=\ann^{i+1}_{z'}$ if and only if $[z]=[z']$.

Let $[z]=[z']$.
If $z$ does not include $\sigma$,
neither does $z'$ and hence both exist in $K_i$.
In this case
$$\ann^{i+1}_z=[\ann^i_z,0]=[\ann^i_{z'},0]=\ann^{i+1}_{z'}$$
since all $p$-simplices other than $\sigma$ gets the same zero element 
appended to their annotations while going from $K_i$ to
$K_{i+1}$. Now consider the case where $z$ includes
$\sigma$. Then, $z'$ also includes $\sigma$.
There is a $p+1$-chain, say $D$, so that $\partial D= z+z'$.
It follows that $\partial D = (z+ \sigma) + (z'+\sigma)$. The $p$-chains
$c=z+\sigma$ and $c'=z'+\sigma$ do not include $\sigma$ since it gets
canceled under $\mathbb{Z}_2$-additions. The $p$-cycle $c+c'=\partial D$
is identity in $H_p(K_i)$ being a boundary. Therefore, its annotation
is zero in $K_i$ giving that $\ann_c=\ann_{c'}$ in $K_i$ and hence
in $K_{i+1}$. It follows that $z=c+\sigma$ and $z'=c'+\sigma$ have
identical annotations in $K_{i+1}$.

Now suppose that $[z]\neq [z']$. If none of $z$ and $z'$ include
$\sigma$, they exist in $K_i$ and by the same logic as above inherit
the same annotations from $K_i$ which cannot be identical because
$K_i$'s annotation is valid. If exactly one of $z$ and $z'$ includes
$\sigma$, the annotation of one in $K_{i+1}$
will have the last element $1$ and that of the
other will have $0$. Thus, they will not be identical. Consider the
remaining case where both $z$ and $z'$ include $\sigma$.
Consider the cycle $z+z'$ which cannot include $\sigma$ because
of $\mathbb{Z}_2$-additions. Then, the cycle $z+z'$ exists in
$K_i$ and cannot be in the class $[0]$ because otherwise
$[z+z']$ will remain identity in $H_p(K_{i+1})$ contradicting $[z]\neq [z']$
in $K_{i+1}$. Since $[z+z']\neq [0]$ in $H_p(K_i)$,
one has $\ann^i_z\neq \ann^i_{z'}$.
It follows that
$$\ann^{i+1}_z=[\ann^i_z,1]\neq [\ann^i_{z'},1]=\ann^{i+1}_{z'}.$$

\noindent Case (ii) of elementary inclusion: The only annotations altered
are for dimensions $p$ and $p-1$. In dimension $p$  the only
change is the addition of $\sigma$ along with its zero
annotation. In this case, $\sigma$ cannot
participate in any $p$-cycle in $K_{i+1}$ because otherwise
$\partial \sigma$ should have zero annotation in $K_i$. Therefore,
annotation for dimension $p$ remains valid in $K_{i+1}$. So, we focus
on dimension $p-1$.

Let $z$ and $z'$ be two $(p-1)$-cycles with
$[z]=[z']$ in $K_{i+1}$. Observe that both $z$ and
$z'$ are also $(p-1)$-cycles in $K_i$.
Recall that $\ann_{\partial\sigma}$ has been added to
all $(p-1)$-simplices with $u$th element equal to $1$. Hence,
the $u$th element of any $(p-1)$-cycle is exactly equal to the
parity of the number of $(p-1)$-simplices in it with $u$th element equal to $1$.
If $[z]=[z']$ in $K_i$, we have $\ann^i_z=\ann^i_{z'}$.
In particular, the $u$th element of $\ann^i_z$ and $\ann^i_{z'}$ are
the same implying that
$\ann_{\partial\sigma}$
has been added with the same parity to $\ann_z^i$ and $\ann_{z'}^i$.

Therefore, $\ann_z^{i+1}=\ann_{z'}^{i+1}$.
Consider the other case when $[z]\neq [z']$ in $K_i$.
Then, there must be a $p$-chain $D$ in $K_i$ such that
$\partial(D+\sigma)=z+z'$ in $K_{i+1}$.
We get $\partial D = z+z'+\partial\sigma$ and hence
$z+z'+\partial\sigma =0$ in $K_i$. So, the annotation of
the cycle $z+z'+\partial\sigma$ is zero in $K_i$.
Since $u$th element of $\ann^i_{\partial\sigma}$
is $1$, it must be true that $\ann^i_z$ and $\ann_{z'}^i$
differ in the $u$th element which means
$$
\ann^{i+1}_z +\ann^{i+1}_{z'}=
\ann^i_z + \ann^i_{z'} +\ann^i_{\partial\sigma}=0.
$$

Now suppose that $[z]\neq [z']$ in $K_{i+1}$. Clearly,
$[z]\neq [z']$ even in $K_i$ implying $\ann^i_z\neq \ann^i_{z'}$.
If $u$th elements of $\ann^i_z$ and $\ann^i_{z'}$ are the same,
we will have
$$\ann^{i+1}_z=\ann_z^i+\ann_{\partial\sigma}^i\neq
\ann^i_{z'}+\ann_{\partial\sigma}^i=\ann^{i+1}_{z'}
$$
proving the required. So, assume that $u$th elements of
$\ann_z^i$ and $\ann_{z'}^i$ are different. Without
loss of generality, assume that $u$th element of
$\ann_z^i$ is $1$ and that of $\ann_{z'}^i$ is $0$.
We claim that
$\ann_z^i+\ann_{\partial\sigma}^i\neq
\ann^i_{z'}$.
Suppose not. Then, by definition of annotation,
$[z+\partial\sigma]=[z']$ in $K_i$. Since $[\partial\sigma]=[0]$
in $K_{i+1}$, we have $[z+\partial\sigma]=[z]=[z']$ in
$K_{i+1}$ reaching a contradiction that $[z]\neq [z']$ in $K_{i+1}$.
Therefore, we have $\ann_z^{i+1}\neq \ann^{i+1}_{z'}$ because
$$
\ann^{i+1}_z=\ann_z^i+\ann_{\partial\sigma}^i\neq
\ann^i_{z'}=\ann^{i+1}_{z'}.
$$

\noindent
Case for elementary collapse: We already know that $f_i$ in this
case is a composition of an inclusion
$i: K_i\hookrightarrow \hat{K_i}$ and a collapse
$f_i':\hat{K_i}\rightarrow K_{i+1}$.
Since we have argued already that our updates under inclusions

maintain valid annotations, we only show that the collapse
under $f_i'$ also does so.

Recall that $f_i'$ is implemented with an annotation transfer followed
by the actual collapse.
Let $\sigma$ be a $p$-simplex where our algorithm adds
its annotation to all other $p$-simplices containing
a simplex $\tau$ that is a $(p-1)$-face of $\sigma$ adjoining $u$.
Adding $\ann_{\sigma}$ to all cofaces of $\tau$ of
codimension $1$ creates
a new annotation which is still valid for $K_i$
by Proposition~\ref{annot-push}.
At the end of all annotation transfers for all $\sigma$,
we have a valid annotation for $K_i$ with the same cohomology basis
such that all vanishing simplices have zero annotation, and each pair of mirror simplices have the same annotation.

Observe that, under the collapse
$\hat{K}_i\stackrel{f_i'}{\rightarrow} K_{i+1}$,
the set of vanishing simplices are exactly those simplices $\sigma$
for which $f_i'(\sigma)$
has a lower dimension than $\sigma$.
A pair of mirror simplices $\tau$ and $\tau'$ are those that
satisfy that $f_i'(\tau) = f_i'(\tau')$ (i.e, the simplex $\tau'$
containing $v$ coincides
with its mirror partner $\tau$ containing $u$).
Hence after the collapse,
if $f_i'(\sigma)$ is a $p$-simplex
for any $p$-simplex $\sigma\in \hat{K}_i$,
we have $\ann_{\sigma}=\ann_{f_i'(\sigma)}$ by construction.
We can now finish the argument that
this induced annotation for $K_{i+1}$ is valid.

Let $z$ and $z'$ be any two $p$-cycles in $K_{i+1}$. Let $w$ and $w'$
be two $p$-cycles in $\hat{K}_i$ so that $f_i'(w)=z$ and $f_i'(w')=z'$.
Then, $[w]=[w']$ if and only
if $[z]=[z']$ since $f'_{i_*}: H_p(K_i)\rightarrow H_p(K_{i+1})$
is an isomorphism (Proposition~\ref{lnk}).
With the modified annotation of $\hat{K}_i$ we have $\ann^i_w=\ann^i_{w'}$
if and only if $[w]=[w']$. Therefore, $[z]=[z']$ in $K_{i+1}$
if and only if $\ann^i_w=\ann^i_{w'}$. The only simplices where
$z$ and $w$ differ are either vanishing simplices or mirror
simplices. In the first case, the annotation is zero and in the second
case the annotations are the same. So, $\ann_w^i=\ann_z^{i+1}$.
Similarly, $\ann_{w'}^i=\ann_{z'}^{i+1}$. Therefore,
$\ann_z=\ann_{z'}$ if and only if $[z]=[z']$ in $K_{i+1}$.
This proves that the annotation for $K_{i+1}$ is valid.\\

\paragraph{Proof of Case (ii) of Proposition~\ref{element-just}.}
In this case, a $(p-1)$-cycle is killed as we add
$\sigma$, so in the reverse direction a cocycle is created.
As before, assume that $[z_1],\ldots,[z_g]$ be a homology
basis for $H_{p-1}(K)$.
By assumption, the $u$th element in the annotation
has been zeroed out. Let $\phi$ be the cocycle
given by $f$ and $\phi_i'$ where $i\in \{1,\ldots,u-1\}$.
Then, as before we get
\begin{eqnarray*}
[\phi(z_1),\ldots,\phi(z_g)]
&=&[\phi_i'(f_\#(z_1),\ldots, \phi_i'(f_\#(z_g))]\\
&=&[\phi_i'(z_1),\ldots,\phi_i'(z_g)]
\end{eqnarray*}
Consider any entry $\phi_i'(z_j)$ in the last vector.
If $\phi_i(z_j)$ has $u$th element $0$, then we must have
$\phi_i'(z_j)=\phi_i(z_j)$. This is because, in that case,
$z_j$ has even number of simplices whose annotations have $u$th
element $1$. Then, according to the update algorithm the annotation
$\ann_{\partial\sigma}$ is added to the simplices
in $z_j$ only even number of times in total maintaining
$\phi_i'(z_j)=\phi_i(z_j)$.

If $\phi_i(z_j)$ has $u$th element $1$, we consider the cycle
$z_j+\partial\sigma$ and observe that $[z_j+\partial\sigma]=[z_j]$
in $H_{p-1}(K')$. Then, $\phi'_i(z_j)=\phi_i'(z_j+\partial\sigma)$
since $\phi_i'$ is derived from a valid annotation for $K'$.
The cycle $z_j$ has odd number of simplices whose annotations
have $u$th element $1$ as $\phi_i(z_j)$ has $u$th element $1$.
So, $\ann_{\partial\sigma}$ has been added odd number of times
to $\ann_{z_j}$ and hence even number of times to
$\ann_{z_j+\partial\sigma}$. This implies that $\phi_i'(z_j+\partial\sigma)
= \phi_i(z_j)$ which leads to $\phi_i'(z_j)=\phi_i(z_j)$.
This immediately gives
$
[\phi(z_1),\ldots,\phi(z_g)]
=[\phi_i'(z_1),\ldots,\phi_i'(z_g)]
$ which we are required to prove.

\section{The Size of $\Rips^{\alpha(1+\eps)^k} (V_k)$}
\label{appendix:size}

%\paragraph{The size of $\Rips^{\alpha(1+\eps)^k} (V_k)$.}
\vspace*{0.1in}
We argue that we can construct every $V_k$ in such way
that each $\Rips^{\alpha(1+\eps)^k} (V_k)$ is of size linear in $n$.
We compute $V_{k+1}$ such that it is a $\delta$-net of $V_k$ for $\delta = \frac{\alpha \eps^2}{2} (1+\eps)^{k-1}$ by the following standard greedy approach:
Let $D(\cdot, \cdot)$ denote the metric on the set of input points $P$ (and thus $V_k$s).
Starting with $V_{k+1} = \emptyset$, pick an arbitrary vertex from $V_k$ and add it to $V_{k+1}$.
In the $i$th round, there are already $i$ points in $V_{k+1}$.
We identify the point $u$ from $V_k$ whose minimum distance
to points in $V_{k+1}$ is the largest.
We stop when either $D(u, V_{k+1}) \le \delta$ or $V_{k+1} = V_k$.
By construction, when this process terminates, any point in $V_k$ is within $\delta$ distance to some point in $V_{k+1}$, and no two points in $V_{k+1}$ are within $\delta$ distance.
A na\"{i}ve implementation of the above procedure takes $O(n^2)$ time. One can also compute the $\delta$-net $V_{k+1}$ more efficiently in $O(n\log n)$ time (see, e.g, \cite{HM06}). However, we remark that this step does not form a bottleneck in the time complexity as computing persistence diagrams takes time cubic in the number of simplices.

\begin{proposition}
Suppose the set of input points $V$ are from a metric space with doubling dimension $d$. For $V_k$s constructed as described above, the number of $p$-simplices in $\Rips^{\alpha(1+\eps)^{k+1}} (V_{k+1})$ is $O((\frac{1}{\eps})^{O(dp)} \cdot n)$ for $0\le \eps \le 1$.
\label{lem:size}
\end{proposition}
\begin{proof}
For simplicity, set $r := \alpha(1+\eps)^{k+1}$; note that $\delta = \frac{\eps^2}{2(1+\eps)^{2}} r$.
We first prove that there are $O((\frac{1}{\eps})^{O(d)})$ number of edges
for each vertex in $V_{k+1}$.
Specifically, consider a node $u \in V_{k+1}$: it will be connected to all other vertices in $V_{k+1}$ that are within distance $r$ to $u$.
Since $V_{k+1}$ is a $\delta$-net of $V_k$, every node in $V_{k+1}$ has a ball centered at it with radius $\delta/2$ that is empty of other points in $V_{k+1}$. 
Since the points are from a metric space with doubling dimension $d$, we can pack only $O( (\frac{r}{\delta/2})^{d}) = O((\frac{4(1+\eps)^{2}}{\eps^2})^d) = O((\frac{4}{\eps^2})^d)$
(for $0\le \eps \le 1$) number of balls of radius $\delta/2$ in a ball of radius $r$.
This means that there are only $O((\frac{1}{\eps})^{2d})$ number of edges
containing $u$, where the big-$O$ notation hides terms exponential in $d$. It then follows that the number of $p$-simplices containing $u$ is $O((\frac{1}{\eps})^{2dp})$. Since there are $|V_{k+1}| \le n$ number of vertices in $V_{k+1}$, the total number of $p$-simplices is bounded by $O((\frac{1}{\eps})^{O(dp)} n)$ as claimed.
\end{proof}


\begin{thebibliography}{99}
\bibitem{ALS11} D. Attali, A. Lieutier, and D. Salinas.
Efficient data structure for representing and simplifying
simplicial complexes in high dimensions.
{\em Proc. 27th Annu. Sympos. Comput. Geom.} (2011), 501--509.

\bibitem{BDM13} J.-D. Boissonnat, T. K. Dey, and C. Maria.
The compressed annotation matrix: An efficient data structure
for computing persistent cohomology. 
{\em Proc. European Sympos. Algorithms} (2013), Vol. 8125
Lecture Notes in Comput. Sci., 695--706.

\bibitem{BD13} D. Burghelea and T. K. Dey. 
\newblock Topological persistence for 
circle-valued maps. 
\newblock {\em Discrete Comput. Geom.}, {\bf 50(1)} (2013), 69--98.

\bibitem{BCCDW12}
O. Busaryev, S. Cabello, C. Chen, T. K. Dey, and Y. Wang.
Annotating simplices with a homology basis and its applications.
{\em Proc. 13th Scandinavian Sympos. Workshops Algorithm Theory (SWAT 2012)},
(2012), 189--200.
%Earlier arxiv version arXiv:1107.3793v2

\bibitem{Carl09}
G.~Carlsson.
\newblock Topology and data.
\newblock {\em Bull. Amer. Math. Soc.} {\bf 46} (2009), 255--308.

\bibitem{CCGGO09}
F. Chazal, D. Cohen-Steiner, M. Glisse, L. J. Guibas, and S. Y. Oudot.
Proximity of persistence modules and their diagrams.
{\em Proc. 25th Annu. Sympos. Comput. Geom.} (2009), 237--246.

%\bibitem{CGOS09} F. Chazal, L. J. Guibas, S. Y. Oudot, and P. Skraba.
%Analysis of scalar fields over point cloud data.
%{\em Proc. 20th ACM-SIAM Sympos. Discrete Algorithms} (2009), 1021--1030.

\bibitem{CK11}
C. Chen and M. Kerber.
An output-sensitive algorithm for persistent homology.
{\em Proc. 27th Annu. Sympos. Comput. Geom.} (2011), 207--216.

%\bibitem{CO08} F. Chazal and S. Oudot. Towards persistence-based
%reconstruction in Euclidean spaces. {\em Proc. 24th Ann. Sympos.
%Comput. Geom.} (2008), 232--241.

\bibitem{CS10} G. Carlsson and V. de Silva. Zigzag persistence.
{\em Found. Comput. Math.}  {\bf 10} (4), 367--405, 2010.

\bibitem{CSM09} G. Carlsson, V. de Silva, and D. Morozov.
Zigzag persistent homology and real-valued functions.
{\em Proc. 26th Annu. Sympos. Comput. Geom.} (2009), 247--256.

\bibitem{Cla06}
K. L. Clarkson.
Nearest-Neighbor Searching and Metric Space Dimensions.
In G. Shakhnarovich, T. Darrel, and P. Indyk, editors, {\em Nearest-Neighbor Methods for Learning and Vision: Theory and Practice}, (2006), 15--59.

\bibitem{CEM06} D. Cohen-Steiner, H. Edelsbrunner, and D. Morozov.
Vines and vineyards by updating persistence in linear time.
{\em Proc. 22nd Annu. Sympos. Comput. Geom.} (2006), 119--126.

\bibitem{CEH07} D. Cohen-Steiner, H. Edelsbrunner, and J. L. Harer.
Stability of persistence diagrams. {\em Discrete Comput. Geom.}
{\bf 37} (2007), 103-120.

\bibitem{DFW13} T. K. Dey, F. Fan, and Y. Wang.
Graph Induced Complex on Point Data.
{\em Proc. 29th Annu. Sympos. Comput. Geom.} (2013),
107--116. 

\bibitem{DSW11}
T. K. Dey, J. Sun, and Y. Wang.
Approximating cycles  in a shortest basis of the first homology group
from point data. {\em Inverse Problems} {\bf 27} (2011),
124004. doi:10.1088/0266-5611/27/12/124004.

\bibitem{DEGN99}
T. K. Dey, H. Edelsbrunner, S. Guha and D. Nekhayev.
Topology preserving edge contraction.
Publications de l' Institut Mathematique (Beograd) {\bf 60(80)} (1999), 23--45.

\bibitem{DMV11} V. de Silva, D. Morozov, and M. Vejdemo-Johansson.
Persistent cohomology and circular coordinates.
{\em Discrete Comput. Geom.} {\bf 45 (4)} (2011), 737--759.

\bibitem{DMV11b} V. de Silva, D. Morozov, and M. Vejdemo-Johansson.
Dualities in persistent (co)homology .
{\em Inverse Problems.} {\bf 27 (12)} (2011), 124003.

\bibitem{EH09}
H.~Edelsbrunner and J.~Harer.
\newblock {\em Computational Topology: {An} Introduction}.
\newblock Amer. Math. Soc., Providence, Rhode Island, 2009.

\bibitem{ELZ02} H. Edelsbrunner, D. Letscher, and A. Zomorodian.
Topological persistence and simplification. {\em Discrete
Comput. Geom.} {\bf 28} (2002), 511--533.

%\bibitem{EN11} J. Erickson and A. Nayyeri. Minimum cuts and shortest
%non-separating cycles via homology covers. {\em Proc. ACM-SIAM Sympos.
%on Discrete Algorithms (SODA)} (2011), 1166--1176.

\bibitem{Ghrist} R. Ghrist.
Barcodes: The persistent topology of data.
{\em Bull. Amer. Math. Soc.} {\bf 45} (2008), 61-75.

\bibitem{HM06}
S. Har-Peled and M. Mendel.
Fast construction of nets in low dimensional metrics, and their applications.
{\em SIAM Journal on Computing}. {\bf 35(5)} (2006), 1148--1184.

\bibitem{Hatcher}
A. Hatcher. Algebraic Topology. Cambridge U. Press, New York, 2002.

\bibitem{NMS11} N. Milosavljevi\'{c}, D. Morozov, and P. \v{S}kraba.
Zigzag persistent homology in matrix multiplication time.
{\em Proc. 27th Annu. Sympos. Comput. Geom.} (2011), 216--225.

\bibitem{Sheehy}
D. Sheehy. Linear-Size Approximations to the Vietoris-Rips Filtration.
{\em Proc. 28th. Annu. Sympos. Comput. Geom.} (2012), 239--247.

\bibitem{ZC05} A. Zomorodian and G. Carlsson. Computing persistent
homology. {\em Discrete Comput. Geom.} {\bf 33} (2005), 249--274.

\end{thebibliography}
\end{document}